\newif \ifIEEE \IEEEtrue
\newif \ifCCS \CCStrue
\newif \ifUSENIX \USENIXtrue

\newif \iffull \fulltrue

\IEEEfalse
\CCSfalse
\USENIXfalse
\fulltrue

\iffull
\documentclass[11pt]{article}
\usepackage[utf8]{inputenc}
\usepackage[letterpaper,margin=1in,bmargin=1.5in]{geometry}
\usepackage{enumitem}
\fi

\ifIEEE
\documentclass[letterpaper,conference]{IEEEtran}
\usepackage[noadjust]{cite}
\usepackage{etoolbox}
\makeatletter
    \patchcmd{\@IEEEsectpunct}{:}{}{}{}
    \renewcommand{\paragraph}{\@startsection{paragraph}{4}{\z@}{1ex \@plus 0.1ex \@minus 0.1ex}{0em}{\normalfont\normalsize\bfseries\boldmath}}
    \renewcommand{\section}{\@startsection{section}{1}{\z@}{1.5ex plus 1.5ex minus 0.5ex}{1ex plus .2ex}{\normalfont\normalsize\centering\scshape}}
    \renewcommand{\subsection}{\@startsection{subsection}{2}{\z@}{1.5ex plus 1.5ex minus 0.5ex}{1ex plus .2ex}{\normalfont\normalsize\itshape}}%
\makeatother

\usepackage{enumitem}
\setlist[itemize]{itemsep=1ex}
\setlist[enumerate]{itemsep=1ex}
\fi

\ifCCS
\documentclass[sigconf]{acmart}
\setcopyright{none}
\acmConference[Anonymous Submission to ACM CCS 2021]{ACM Conference on Computer and Communications Security}{Due Jan 20 2021}{Seoul, South Korea}
\acmYear{2021}

\makeatletter
\renewcommand\paragraph{\@startsection{paragraph}{4}{\z@}%
  {-.5\baselineskip \@plus -1.5\p@ \@minus -1\p@}%
  {-3.5\p@}%
  {\ACM@NRadjust{\normalfont\normalsize\bfseries\boldmath\@adddotafter}}}
\makeatother

\usepackage{enumitem}
\fi

\ifUSENIX
\documentclass[letterpaper, twocolumn, 10pt]{article}
\usepackage{usenix}
\usepackage{enumitem}
\setlist[itemize]{noitemsep}
\setlist[enumerate]{noitemsep}
\fi

\newif \ifcomments \commentstrue

\usepackage[dvipsnames]{xcolor}

\usepackage{amsmath,amsfonts}
\usepackage{amsthm}
\usepackage{mathrsfs}
\usepackage[]{algorithm2e}
\RestyleAlgo{boxruled}
\usepackage{booktabs}
\LinesNumbered
\usepackage{enumitem}
\usepackage{setspace}

\iffull
\usepackage[style=numeric-comp,backend=bibtex]{biblatex}
\usepackage[pdfauthor={},pdfpagelabels=true,linktocpage=true,hidelinks,linktoc=all,colorlinks]{hyperref}
\usepackage[nameinlink,capitalize,noabbrev]{cleveref}
\hypersetup{linkcolor = {black}, citecolor = {magenta}, urlcolor = {black}}
\else
\usepackage{hyperref}
\fi

\usepackage{xparse}
\usepackage{array,framed}
\usepackage{caption}
\usepackage{subcaption}
\usepackage{tikz,pgfplots,pgfplotstable}
\usetikzlibrary{shapes.geometric,arrows,external,pgfplots.groupplots,matrix}

\usepackage{thmtools}
\usepackage{thm-restate}

\usepackage{mathtools}

\DeclareMathAlphabet{\mathcal}{OMS}{cmsy}{m}{n}

\ifIEEE
\usepackage[scaled=0.88]{helvet} 
\fi

\iffull 
  \newcommand{\fig}{Figure\xspace}
\else
  \ifIEEE
    \newcommand{\fig}{Fig.\xspace} 
  \else
    \newcommand{\fig}{Figure\xspace}
  \fi
\fi

\iffull 
  \newcommand{\figs}{Figures\xspace}
\else
  \ifIEEE
    \newcommand{\figs}{Fig.\xspace} 
  \else
    \newcommand{\figs}{Figures\xspace}
  \fi
\fi

\ifcomments
    \newcommand{\ari}[1]{{\small\textsf{\color{blue}{[Ari: {#1}]}}}}

    \newcommand{\mahimna}[1]{{\small\textsf{\color{violet}{[Mahimna: {#1}]}}}}
    \newcommand{\phil}[1]{{\small\textsf{\color{olive}{[Phil: {#1}]}}}}
    \newcommand{\kushal}[1]{{\small\textsf{\color{cyan}{[Kushal: {#1}]}}}}
    \newcommand{\todo}[1]{{\small\textsf{\color{red}{[ToDo:  {#1}]}}}}
    \newcommand{\assigned}[1]{{\small\textsf{\color{red}{[Assigned:  {#1}]}}}}

\else
    \newcommand{\ari}[1]{}
    \newcommand{\mahimna}[1]{}
    \newcommand{\phil}[1]{}
    \newcommand{\kushal}[1]{}
    \newcommand{\todo}[1]{}
    \newcommand{\assigned}[1]{}
\fi

\iffull
  \newcommand{\mypara}{\paragraph}
\else
  \newcommand{\mypara}[1]{\smallskip\noindent\textbf{#1}\hspace*{0.5em}}
\fi

\newcommand{\players}{{\mathcal{P}}}
\newcommand{\accounts}{{A}}
\newcommand{\contracts}{{\mathcal{C}}}
\newcommand{\account}{\mathsf{acc}}
\newcommand{\accbalance}{\mathsf{balance}}
\newcommand{\accdata}{\mathsf{data}}
\newcommand{\accdatasize}{d}
\newcommand{\calleracc}{\textsf{acc}_{\textnormal{caller}}}
\newcommand{\token}{\mathbf{T}}
\newcommand{\caller}{\textnormal{caller}}

\newcommand{\txset}{{\mathcal{T}}}
\newcommand{\tx}{\mathsf{tx}}
\newcommand{\txaction}{\mathsf{action}}

\newcommand{\block}{B}
\newcommand{\blockaction}{\mathsf{action}}
\newcommand{\validblocks}{\mathsf{validBlocks}}
\newcommand{\blocknum}{\mathsf{num}}

\newcommand{\statemaps}{{\mathcal{S}}}
\newcommand{\MEV}{\textnormal{MEV}}
\newcommand{\WMEV}{\textnormal{WMEV}}
\newcommand{\EV}{\textnormal{EV}}

\newcommand{\kMEV}{\ensuremath{k}\textrm{-}\MEV}

\newcommand{\uniswapcon}{C_\textnormal{uniswap}}

\newcommand{\betcontract}{C_\textnormal{pricebet}}
\newcommand{\makercon}{C_\textnormal{maker}}
\newcommand{\makeracc}{\makercon}
\newcommand{\collateral}{\textsf{collateral}}
\newcommand{\debt}{\textsf{debt}}
\newcommand{\liquidationratio}{\textsf{threshold}}
\newcommand{\qty}{\textrm{qty}}

\newcommand{\ETH}{\textnormal{ETH}}
\newcommand{\BBT}{\textnormal{BBT}}

\newcommand{\kcell}{\texttt{k}\xspace}
\newcommand{\scell}{\texttt{S}\xspace}
\newcommand{\bcell}{\texttt{B}\xspace}

\newcommand{\secparam}{\lambda}

\newcommand{\bits}{\{0,1\}}

\newcommand{\poly}{\textnormal{\textsf{poly}}}

\newcommand{\true}{\mathsf{true}}
\newcommand{\false}{\mathsf{false}}

\newcommand{\ifcode}{\textbf{if }}
\newcommand{\thencode}{\textbf{then }}
\newcommand{\elsecode}{\textbf{else }}

\newcommand{\andcode}{\textbf{and }}

\newcommand{\functioncode}{\textbf{function }}

\newcommand{\Z}{{\mathbb{Z}}}

\newcommand{\pluseq}{\mathrel{+}=}
\newcommand{\minuseq}{\mathrel{-}=}

\newcommand{\myind}{\hspace*{1em}}

\newcommand{\stretchval}{1.2}
\newcommand{\fpage}[2]{\begin{center}\framebox{\begin{minipage}{#1\columnwidth} \setstretch{\stretchval} \footnotesize #2 \end{minipage}}\end{center}}
\newtheorem{theorem}{Theorem}
\newtheorem{lemma}[theorem]{Lemma}

\theoremstyle{definition}
\newtheorem{definition}[theorem]{Definition}
\newtheorem{example}{Example}

\newtheorem{remark}{Remark}
\newtheorem{characteristic}{Characteristic}

\iffull
\newcommand{\triplefigwidth}{0.45\textwidth}

\else
\newcommand{\triplefigwidth}{0.68\columnwidth}
\fi

\usepackage{listings}
\usepackage{color}
 
\definecolor{dkgreen}{rgb}{0,0.6,0}
\definecolor{gray}{rgb}{0.5,0.5,0.5}
\definecolor{mauve}{rgb}{0.58,0,0.82}
\definecolor{gray}{rgb}{0.4,0.4,0.4}
\definecolor{darkblue}{rgb}{0.0,0.0,0.6}
\definecolor{lightblue}{rgb}{0.0,0.0,0.9}
\definecolor{cyan}{rgb}{0.0,0.6,0.6}
\definecolor{darkred}{rgb}{0.6,0.0,0.0}

\definecolor{lightgray}{rgb}{0.9,0.9,0.9}

\usepackage{listings}
\lstdefinelanguage{XML}
{
  morestring=[s][\color{mauve}]{"}{"},
  morestring=*[s][\color{black}]{>}{<},
  morestring=*[s][\color{purple}]{~>},
  morestring=*[s][\color{orange}]{+Int},
  morestring=*[s][\color{orange}]{-Int},
  morestring=*[s][\color{orange}]{*Int},
  morestring=*[s][\color{orange}]{/Int},
  morestring=*[s][\color{black}]{<}{>},
  morestring=*[s][\color{black}]{</}{>},
  morecomment=[s]{<?}{?>},
  morecomment=[s][\color{dkgreen}]{<!--}{-->},
  stringstyle=\color{black},
  identifierstyle=\color{black},
  keywordstyle=\color{red},
  morekeywords={in,var,exec,gets,xmlns,xsi,noNamespaceSchemaLocation,type,id,x,y,source,target,version,tool,transRef,roleRef,objective,eventually}
}
\iffull \bibliography{references} \fi

\ifIEEE
\setitemize{leftmargin=*}
\fi

\makeatletter
\def\blfootnote{\xdef\@thefnmark{}\@footnotetext}
\makeatother

\begin{document}
\blfootnote{A preliminary version of this paper appears in the proceedings of IEEE S\&P 2023. This is the full version.\\}

\title{Clockwork Finance: Automated Analysis of Economic Security in Smart Contracts}

\iffull
\author{Kushal Babel\thanks{The first three authors contributed equally to this work.} \and Philip Daian\footnotemark[1] \and Mahimna Kelkar\footnotemark[1] \and Ari Juels}
\date{
{Cornell Tech, Cornell University, and IC3} \\
}
\fi

\ifCCS \begin{abstract}
We introduce the {\em Clockwork Finance Framework} (CFF), a general purpose, formal verification framework for mechanized reasoning about the economic security properties of  {\em composed decentralized-finance (DeFi) smart contracts}.  


CFF features three key properties.  It is \emph{contract complete}, meaning that it can model any smart contract platform and all its contracts---Turing complete or otherwise. It does so with asymptotically \emph{constant model overhead}. It is also \emph{attack-exhaustive by construction}, meaning that it can automatically and mechanically extract all possible economic attacks on users' cryptocurrency across modeled contracts. 

Thanks to these properties, CFF can support multiple goals: economic security analysis of contracts by developers, analysis of DeFi trading risks by users, fees UX, and optimization of arbitrage opportunities by bots or miners. Because CFF offers composability, it can support these goals with reasoning over any desired set of potentially interacting smart contract models.

We instantiate CFF as an executable model for Ethereum contracts that incorporates a state-of-the-art deductive verifier. Building on previous work, we introduce {\em extractable value} (EV), a new formal notion of economic security in composed DeFi contracts that is both a basis for CFF and of general interest. 

We construct modular, human-readable, composable CFF models of four popular, deployed DeFi protocols in Ethereum: Uniswap, Uniswap V2, Sushiswap, and MakerDAO, representing a combined 24 billion USD in value as of March 2022. We use these models along with some other common models such as flash loans, airdrops and voting to show experimentally that CFF is practical and can drive useful, data-based EV-based insights from real world transaction activity. {\em Without any explicitly programmed attack strategies}, CFF uncovers on average an expected \$56 million of EV per month in the recent past.

\end{abstract}
 \maketitle
\else \maketitle  
 \fi

\iffull
\newpage\tableofcontents\newpage
\hypersetup{linkcolor = {cyan}, citecolor = {magenta}, urlcolor = {black}}
\fi


\section{Introduction}
\label{sec:introduction}

\noindent The innovation of smart contracts has resulted in an explosion of decentralized applications on blockchains. Abstractly, smart contracts are pieces of code that run on blockchain platforms, such as Ethereum. They support rich (even Turing-complete) semantics, can trade in the underlying cryptocurrency, and can directly manipulate blockchain state.
While early blockchains were built primarily to support currency transfer, newer ones with smart contracts have enabled a wide range of sophisticated and novel decentralized applications. 

One particularly exciting area where smart contracts have been influential is decentralized finance (or \textit{DeFi}), a general term for financial instruments built on top of public decentralized blockchains. DeFi contracts have realized a number of financial mechanisms and instruments (e.g., automated market makers~\cite{buterin2017path}, atomic swaps~\cite{van2019specification}, and flash loans~\cite{qin2020attacking}) that cannot be replicated with fiat or real world assets, and have no analog in traditional financial systems. These innovations usually take advantage of two distinctive properties of smart contracts. These are \textit{atomicity}, which means (potential) execution of multi-step transactions in an all-or-nothing manner, and \textit{determinism}, meaning execution of state transitions without randomness and thus a unique transaction outcome for a given blockchain state. Smart contracts can also intercommunicate on-chain, which has led to DeFi instruments that can interoperate and {\em compose} to achieve functionality that transcends their independent functionalities.  

Recent years, however, have seen a plethora of high-profile attacks on DeFi contracts (see, e.g., \cite{chen2020survey} for a recent survey), with attackers stealing billions in the aggregate.  These attacks are primarily financial in nature and not pure software exploits; they leverage complex financial interactions among multiple DeFi contracts whose composition is poorly understood. Existing notions of software security and traditional bug-finding tools are insufficient to reason about or discover such attacks. 

A range of literature~\cite{hildenbrandt2018kevm, hirai2017defining}, has attempted to apply formal verification techniques to the study of DeFi security. These works, though, have typically been used to check for attack heuristics~\cite{zhou2020high} that represent conventional software bugs in smart contracts or to validate formal security properties~\cite{park2020end, krupp2018teether} akin to those in standard software verification tools. 
More recently, some work~\cite{zhou2021just} has applied formal verification tools to the economic security of DeFi contracts, quantifying such security by identifying optimum arbitrage strategies. While an important initial step, this work has focused on predetermined, known attack strategies, and lacks the generality to discover new economic attacks, rule out classes of attacks, or provide upper bounds on the exploitable value of DeFi contracts.

\mypara{Clockwork Finance.}
Motivated by the limited formal exploration of the question of DeFi contracts' economic security, in this paper we present \textit{Clockwork Finance\footnote{Our name comes from the Enlightenment notion of the cosmos as a clock, i.e., a fully deterministic and predictable machine, like the smart contract systems we consider. The Wikipedia definition of {\em clockwork universe}~\cite{clockwork-wiki} notes: ``In the history of science, the clockwork universe compares the universe to a mechanical clock ... making every aspect of the machine predictable.''} (CF)}, an approach to understanding the economic security properties of DeFi smart contracts and their composition. CF addresses the inherently economic nature of DeFi security properties by codifying the use of formal verification techniques to reason about the \textit{profit} extractable from the system by a participant, rather than in terms of more traditional descriptions of software bugs as error states. CF relies on, and we introduce in this paper, the first formal definition for the economic security of composed smart contracts, which we call {\em extractable value} (EV). EV generalizes {\em miner-extractable value} (MEV)---a metric defined in~\cite{daian2020flashboys} to study DeFi protocol impact on consensus security.\footnote{CF can be extended to other metrics of economic security, e.g.,  arbitrageurs' profits, profits of permissioned actors, etc., but we leave extensions to future work.} 

\mypara{Clockwork Finance Framework (CFF).}
We realize CF in the form of a powerful mechanized tool that we call the {\em Clockwork Finance Framework (CFF)}. To use CFF, a user wishing to analyze the economic security of a contract creates or reuses an existing formal model of the contract, as well as models for potentially composed contracts. CFF, together with the models we provide, offers three key functional properties:

\begin{itemize}
\item \underline{\em Contract completeness:}   CFF is \emph{contract complete} in the sense that it can model DeFi (and other) contracts, such as those in Ethereum, with equivalent execution complexity to the native platform.  That is, for all possible transactions (inputs), executing the formal CFF model of a contract requires time $\mathcal{O}(1)$ overhead over EVM/native execution time. CFF introduces no execution blow-up or time penalty for the execution of any transaction sequence, even for complex compositions of contracts. CFF also has \textbf{equal expressive power} as the contract platform to which it's applied---again, such as Ethereum.

\item \underline{\em Constant model overhead:} The models we provide feature at most a (small) constant-size increase in the size (number of distinct semantic paths) of the model compared to the target contract. Oftentimes, with path pruning, our specialized models are even substantially smaller than the smart contract code being modeled. We provide a general approach for achieving this property for new CFF models. We discuss this approach and property in detail in Section~\ref{subsec:equivalence}.


\item \underline{\em Attack-exhaustive by construction:}  CFF can mechanically reason about the full space of possible state transitions for the given set of transactions and models. CFF can in principle---given sufficient computation---identify any attack expressible in our definitions as a condition of mempool transaction activity and target contract models. We ensure this by making sure our provided models are \emph{over-approximations} of the studied contracts, yielding false positives in the attack search as a trade-off for efficiency, but not false negatives. We then prune these false positives through concrete validation. We discuss this property in detail along with sources of unsoundness in Section~\ref{subsec:design_impl}.
\end{itemize}

\noindent CFF also offers two important usability features:

\begin{itemize}

\item \underline{\em Modularity:} CFF models are {\em modular}, meaning that once a model is realized for a particular contract, it can be used for any CFF execution involving that contract. Modularity also means that models are arbitrarily composable in CFF: any and all models in a library can be invoked for a CFF analysis without customization. 

\item \underline{\em Human-readability:} Although we do not show this experimentally, we show by example that CFF models are typically easier for human users to read, understand, and reason about than contract source code.
\end{itemize}

Taken together, these properties and features make CFF  highly versatile and able to support a range of different uses. Designers of DeFi contracts can use CFF to reason about the economic security of their contracts and do so, critically, while reasoning about interactions with other contracts. Arbitrage bots and miners can use the same contract models to find profitable strategies in real-time. Users can use CFF to reason about guarantees provided by the transactions they execute in the network, including the value at risk of exploits by miners, bots, and other network participants---which today is considerable in practice~\cite{daian2020flashboys,zhou2020high}. With the rise of frontrunning-as-a-service~\cite{flashbots}, users can also use CFF to set the right fees for their transactions, which taken together with the value extractable from their transactions determines inclusion in the block. We explore these various use cases in the paper.

CFF achieves more than mere measurement of economic security: It can {\em prove bounds} on the economic security of contracts, i.e., the maximum amount adversaries can extract from them. Furthermore, it can do so using only the formally specified models of interacting contracts. CFF {\em does not require manual coding of adversarial strategies}. 

Notably, this means that CFF can illuminate potential adversarial strategies even when they were \textit{not previously exploited} in the wild. This stands in contrast to existing work, where the focus has often been on specific predefined strategies encoded manually~\cite{zhou2020high}, or which has required error-prone effort to define an action-space manually beyond the mere contract code executing on the system~\cite{zhou2021just}. We believe that use of CFF would be a helpful part of the standard security assessment process for smart contracts, alongside bug finding, auditing, and conventional formal verification.

\iffull
\mypara{Contributions.}

We summarize three concrete contributions and insights from our paper below:

\begin{itemize}[itemsep=3pt]
    \item \textbf{Security Definitions} (Sections~\ref{sec:model} and~\ref{sec:composability}).
    We provide the first formal definitions for the economic security of smart contracts and their composition and thus the first principled basis for DeFi contract designers to reason about the economic security of their protocols. Our definitions are general enough to model different types of players with different capabilities (e.g., transaction reordering, censorship, inserting malicious transactions) for influencing the system state.

    \item \textbf{Clockwork Finance Framework (CFF) and Concrete Models} (Sections~\ref{sec:model} and~\ref{sec:kmodel}). We instantiate our definitions in our CFF tool in order to find arbitrage strategies and prove bounds on the economic security of smart contracts. We model within CFF and analyze four popular real-world contracts: Uniswap V1, Uniswap V2, SushiSwap, and MakerDAO.  We compare our results again direct on-chain Ethereum Virtual Machine execution, showing that CFF execution of our models yields high-fidelity results.

    \item \textbf{Practical Attacks and Formal Proofs} (Section~\ref{sec:experiments}).
    Our CFF tool automatically discovers the main attack patterns seen in practice, uncovering highly profitable attacks in an automated way for the four contracts we model. These attacks exploit the price slippage or the lack of secure financial composition of DeFi contracts, and can be used by malicious miners (or others) to profit at the expense of ordinary users. Our tool also yields formal mathematical proofs for the upper bound on the value extractable from these attacks. By our conservative estimate, the potential impact of these attacks frequently exceeds the Ethereum block reward by two orders of magnitude (i.e., 10,000\%). We also validate our attacks by simulating them on an archive node and have contributed the implementation of our simulation method into the latest public release of the Erigon client software.

\end{itemize}

\vspace{2mm}
\fi

CFF is the first smart-contract analytics tool to achieve contract completeness, constant model overhead, and attack-exhaustiveness by construction, enabling it to bring new capabilities to ecosystem participants.  
\iffull
Complete CFF code is available at \url{https://github.com/defi-formal/cff/}.
\else
Complete CFF code is available at \url{https://github.com/defi-anon/cff/}.
\fi

\iffull \else
The full version of this paper~\cite{fullversion} includes additional analyses and discussions.
\fi




\iffull

\section{Background and Related Work}
\label{sec:background}


Our work intersects with several well-studied areas which we briefly introduce here as background.

\subsection{Blockchain and Smart Contracts}


Smart contracts are executed in \emph{transactions}, which, like ACID-style database transactions~\cite{vossen1995database}, modify the state of a cryptocurrency system atomically (that is, either the entire transaction executes or no component of the transaction executes). A transaction's output and validity depends on both the system's state and the code being executed, which can read and respond to this state. The state may also include user balances of tokens representing assets or of cryptocurrencies in the underlying system. In the smart contract setting, the primary purpose of the underlying blockchain is to order transactions. The execution of a transaction sequence is then deterministic, and can be computed by all parties. The sequencing of transactions is done by actors known as \emph{miners} (or \emph{validators} or \emph{sequencers}, terms we use interchangeably).

A unique attribute of smart contract transactions that proves critical to decentralized finance is their ability to throw an unrecoverable error, reverting any side-effects of a transaction until that point and converting the transaction into a no-op. This allows actors to execute transactions in smart contracts that are reverted if some operation fails to complete as expected or yield desired profit. 


\iffull
\subsection{Decentralized Finance}
Decentralized Finance, or \emph{DeFi}, is a general term for the ecosystem of financial products and protocols defined by smart contracts running on a blockchain. As of August 2021, the Ethereum DeFi space contains roughly 80bn USD of locked capital in smart contracts~\cite{defipulse}. DeFi protocols or instruments have already been deployed for a wide range of use cases, and allow users to borrow, lend, exchange, or trade assets on a blockchain. Abstractly, a key goal of DeFi is to create composable and modular financial instruments that do not rely on a centralized issuing party. DeFi instruments can thus interoperate programatically without human intervention or complex cooperation among issuing entities. We provide a brief background on the two specific classes of DeFi instruments featured in this work.

\mypara{Lending contracts.}
Some DeFi contracts lend a certain cryptocurrency (such as DAI in the Maker protocol~\cite{makerdao-whitepaper}) to a user, with another user-supplied cryptocurrency (such as ETH) held by the contract as collateral. If the value of the collateral falls below a system-defined threshold, the financial instrument can automatically foreclose on the collateral to repay the loan without the cooperation of the borrower. This automated loan guarantee mitigates risk in a way attractive to lenders. Lending contracts can also underlie ``stablecoin" protocols, which support tokens pegged to real-world currencies such as the U.S. dollar (e.g., as in the Maker protocol).

\mypara{Decentralized exchanges.}
Another example of a DeFi instrument is a \emph{decentralized exchange} (``DEX"). In a DEX, users can trade between different assets that have a digital representation (e.g., on a blockchain). A DEX facilitates the exchange of assets without the risk that one party in the exchange defaults or fails to execute their end of the asset swap. This guarantee protects users from counterparty risk present in traditional exchanges, especially cryptocurrency exchanges, which have often violated users' trust assumptions by absconding with funds~\cite{mcmillan2014inside, moore2013beware} or  incorrectly executing user orders through technical errors and even fraud~\cite{twomey2019fraud}. A special class of DEX called Automated Market Maker (``AMM'') eliminates the need for a counterparty to execute a swap. An AMM (like Uniswap or Sushiswap) maintains reserves of liquidity providers' assets and allows swaps with a user's assets at programatically self adjusting prices.

\fi
\mypara{Miner extractable value.} A notion called {\em MEV}, or \emph{miner-extractable value}, introduced in~\cite{daian2020flashboys}, measures the extent to which miners can capture value from users through strategic placement and/or ordering of transactions. Miners have the power to dictate the inclusion and ordering of mempool transaction in blocks. (Thus MEV is a superset of the front-running/arbitrage profits capturable by ordinary users or bots, because miners have strictly more power.) Previous studies of MEV have performed transaction-level measurements of the outcome of specific strategies (e.g., sandwiching attacks in~\cite{zhou2020high} and pure revenue trade composition in~\cite{daian2020flashboys}). Other work has abstracted away transaction-level dynamics, analyzing DeFi protocols such as AMMs using  statistical modeling and economic agent-based simulation~\cite{angeris2020improved}.

\subsection{Formal Verification Tools}
\label{subsec:formal-verification}

Formal verification is the study of computer programs through mathematical models in well-defined logics. It supports the proof of mathematical claims over the execution of programs, traditionally to reason about program safety and correctness. Formal verification has been applied to traditional financial systems in the past (like~\cite{passmore2017tradfi}) but as noted in Section~\ref{sec:introduction}, DeFi systems have novel properties not present in these older systems. Most formal verification works for smart contracts (such as~\cite{amani2018towards, hirai2017defining, zhou2020high, park2020end, arusoaie2019findel}) do not reason about economic security and hence cannot characterize financial exploits in DeFi (i.e., they are not attack-exhaustive by construction). Recent work~\cite{zhou2021just} has attempted to apply formal verification to find profitable arbitrage strategies but does not provide formal proofs of economic security. Moreover, the tool covers only certain types of manually encoded smart contract actions, so that the tool lacks contract completeness and optimal model sizes. 


Our work aims to establish a clear translation interface between existing program verification tools and the unique security requirements of DeFi. We develop our models in the K Framework~\cite{rosu2017k}, which provides a formal semantics engine for analyzing and proving properties of programs. K allows developers to define models that are \emph{mathematically formal}, \emph{machine-executable}, and \emph{human-readable}.

\begin{figure}
\centering
\includegraphics[scale=0.18]{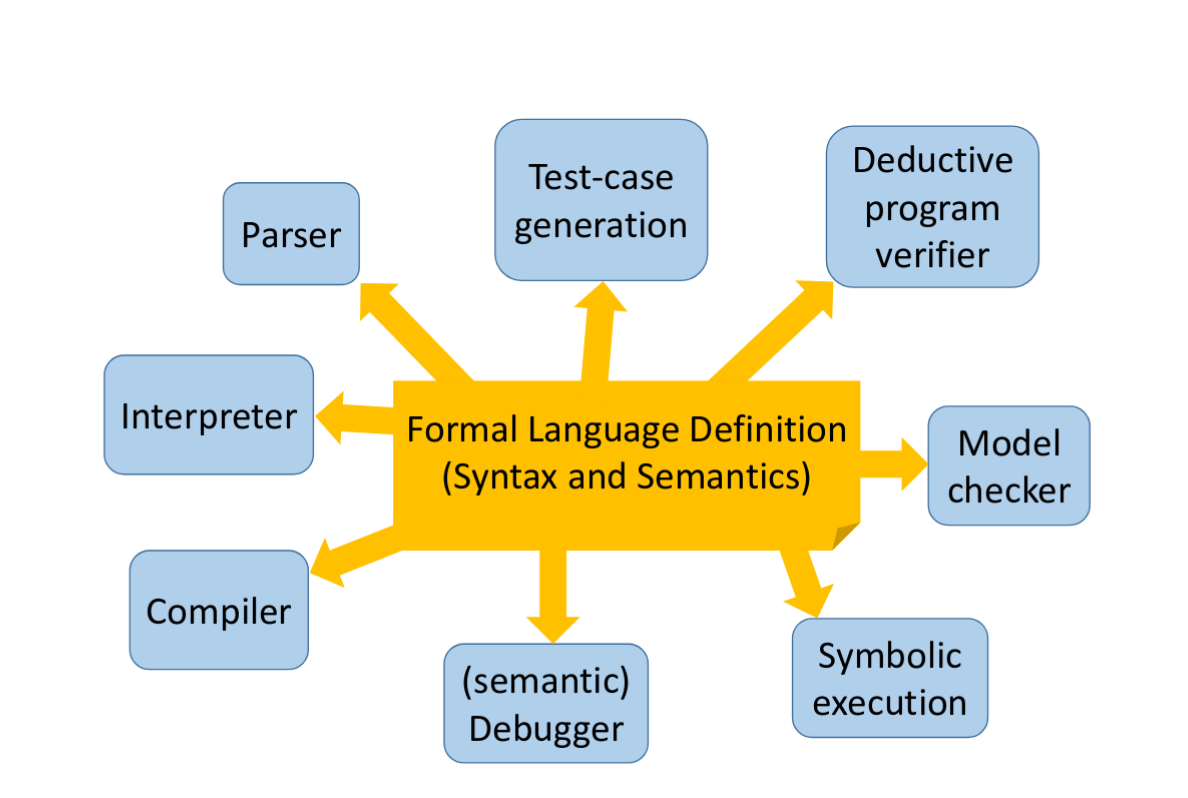}
\caption{K Framework: In this figure from~\cite{rosu2017k}, the yellow box is a user-specified language model (like that in Section~\ref{sec:kmodel}); blue boxes are tools generated automatically by the framework.}
\label{fig:koverview}
\iffull \else \vspace{-2em} \fi
\end{figure}

By {\em mathematically formal}, we mean that K uses an underlying theory called ``matching logic'' that allows claims expressed about programs in programming languages defined by K to be proven formally. Such proofs have been used in industry to verify the practical security properties of smart contracts that hold billions of dollars~\cite{verifiedsc}.


By {\em executable}, we mean that K provides concurrent and non-deterministic rewrite semantics~\cite{chen-rosu-2018-isola} that allow for efficient execution of large programs in the developer-specified programming language model. \fig~\ref{fig:koverview} shows the high-level goals of the K Framework, which include deriving an interpreter and compiler for a specified language semantics, as well as model-checking tools.

By {\em human-readable}, we mean that K provides output in a form that can serve as a reference for other mathematical models, as it uses only abstract and human-readable mathematical operations. Examples of human-readable K semantics include the Jello Paper for the Ethereum VM.\footnote{The ``Jello Paper" (\url{https://jellopaper.org/}), based on~\cite{hildenbrandt2018kevm}, reimplements the original Ethereum yellow paper~\cite{wood2014ethereum} in a machine executable, mathematically formal manner and can generate an Ethereum interpreter and contract proofs.} Because DeFi contracts today lack standardized abstract models, we believe K's abstract models are especially suitable to DeFi and hope they can ease security analysis and specification.

K is one of a number of formal verification tools; other common tools include Coq, Isabelle, etc. Indeed, several have been applied to model Ethereum-based systems in the past~\cite{amani2018towards, hirai2017defining, arusoaie2019findel}. We refer the reader to~\cite{rosu2017k,chen-lucanu-rosu-2020-trb,chen-rosu-2018-isola} for details on the mathematical and formal foundations of K. We emphasize that our MEV-based secure composability definitions and general results are not specific to K.



\fi
\section{Clockwork Finance Formalism}
\label{sec:model}

We introduce our formalism for Clockwork Finance in this section. It underpins the definition of {\em extractable value} (\EV) we introduce in this paper. Our contract composability definitions in Section~\ref{sec:composability} are  based in turn on that of \EV.
We let $\secparam$ throughout denote the system security parameter.

\iffull \else
We assume the reader has general background knowledge on software formal verification outside of cryptocurrencies, and on some basic cryptocurrency and smart contract concepts. We provide further background for readers who do not share this context in Appendix~\ref{sec:background}.
\fi

\mypara{Accounts and balances.}
We use $\accounts$ to denote the space of all possible accounts. For example, in Ethereum, accounts represent public key identifiers and are 160-bit strings (in other words, $\accounts = \bits^{160}$).
We define two functions, $\accbalance$: $\accounts \times \token \to \Z$ and $\accdata$: $\accounts \to \bits^{\accdatasize}$ (where $d$ is $\poly(\secparam)$), that map an account to its current balance (for a given token $\token$) and its associated data (e.g., storage trie in Ethereum) respectively. For $a \in \accounts$, as shorthand, we let $\accbalance(a)$ denote the balance of all tokens held in $a$ and $\accbalance(a)[\token]$ denote the account balance of token $\token$.  We use $\accbalance(a)[0]$ denote the balance of the primary token (e.g., ETH in Ethereum\footnote{We note that our usage of \textit{token} to denote ETH is non-standard. While the ETH balance is stored differently than the balance of other tokens in Ethereum, we choose to model them using the same $\accbalance$ function for a cleaner (although equivalent) formalism.}).


We define the current {\em system state mapping} (or simply {\em state}) $s$ as a combination of the account balance and data; that is, for an account $a$, $s(a) = (\accbalance(a), \accdata(a))$. We use $\statemaps$ to denote the space of all state mappings.

\mypara{Smart contracts.}
As smart contracts in the system are globally accessible, we model them within the global state through the special $\mathsf{0}$ account. We let $\contracts(s)$ denote the set of contracts in state $s$ of the system, which may change as new contracts are added. We use $\accbalance(C,s)$ and $\accdata(C,s)$ to denote the balance of tokens and the data (e.g., contract state and code) associated with a contract $C$ in state $s$ respectively.

\mypara{Transactions.} 
Transactions are polynomial-sized (in the security parameter) strings constructed by some player that are executed by the system and can change the system state. Abstractly, a transaction $\tx$ can be represented by its \textit{action}: a function from $\statemaps$ to $\statemaps \cup \{\bot\}$ transforms the current state mapping into a new state mapping. We denote this action function by $\txaction(\tx)$. We say that a transaction $\tx$ is valid in state $s$ if $\txaction(\tx)(s) \neq \bot$ and use $\txset_s$ to denote the set of all valid transactions for state $s$. Our formalism is general enough to also allow transactions that add smart contracts to the system or interact with existing ones.



\mypara{Blocks.}
We define a block $\block = [\tx_1, \dots, \tx_l]$ to be an ordered list of transactions. We disregard block contents regarding consensus mechanics, e.g., nonce, blockhash, Merkle root which are not relevant for our framework. Of the block metadata, we only model the block number, denoted by $\blocknum(\block)$. The action of a block can now be defined as the result of the action of the sequence of transactions it contains. We use $\blockaction(\block)(s)$ to denote the state resulting from the action of $\block$ on starting state $s$. That is, $\blockaction(\block)(s) = \txaction(\tx_l)(s_{l-1})$ where $s_0 = s$ and $s_i = \txaction(\tx_i)(s_{i-1})$. A block is said to be valid if all of its transactions are valid w.r.t.~their input state (i.e., the state resulting from executing prior transactions sequentially).

We can analogously define the action of any sequence of transactions (even spanning multiple blocks)---a concept useful for analyzing reordering across blocks.   

\mypara{Network actors and mempools.}
Let $\players$ denote the (unbounded) set of players in our system, and $P \in \players$ denote a specific player. 
We use $\txset_s$ to denote the global set of all valid transactions for state $s$, but note that not all transactions can be validly generated by all players.
For a player $P \in \players$, we define a set $\txset_{P,s} \subseteq \txset_s$ as the transactions that can be validly created by $P$ when the system is in state $s$. Transactions created by players are included in a mempool for the current state. A player $P$ working as a miner to create a block may include any transactions currently in the mempool (i.e., transactions generated by other players) as well as any transactions in $\txset_{P,s}$ that $P$ generates itself. Note however, that the miner cannot change the contents of other players' transactions, as they are digitally signed. Abstractly, a ``valid block'' for a miner is any sequence of transactions that the miner has the \textit{ability} to include. We use $\validblocks(P, s)$ to denote the set of all valid blocks that can be created by player $P$ in state $s$ if it could work as a miner. We use $\validblocks_k(P, s)$ to denote the set of valid $k$ length block sequences $(\block_1, \cdots, \block_k)$ such that $\block_1 \in \validblocks(P, s)$ and the other $\block_i \in \validblocks(P, s_{i-1})$ where $s_0 = s$ and $s_j = \blockaction(\block_j, s_{j-1})$.

\mypara{Extractable value.}
Equipped with our basic formalism, we now define extractable value ($\EV$), which intuitively represents the maximum value, expressed in terms of the primary token, that can be extracted by a given player from a valid sequence of blocks that extends the current chain. Formally, for a state $s$, and a set $\mathcal{B}$ of valid block sequences of length $k$, the $\EV$ for a player $P$ with a set of accounts $A_P$ is given by:
\begin{align*}
 \EV(P, \mathcal{B}, s) = \max_{(\block_1, \dots, \block_k) \in \mathcal{B}} \left\{ \sum_{a \in A_P}
 \begin{array}{l}
 \accbalance_k(a)[0] \\
 - \accbalance_0(a)[0]
 \end{array}
 \right\}.
\end{align*}
where $s_0 = s = (\accbalance_0, \accdata_0)$, $s_i = \blockaction(\block_i)(s_{i-1})$, and $s_k = (\accbalance_k, \accdata_k)$.



We also define miner-extractable value, which computes the maximum value that a \textit{miner} can extract in a state $s$. Consider a player $P$ working as a miner.
The $\kMEV$ of $P$ in state $s$ can now be defined as:
\[
    \kMEV(P,s) = \EV(P, \validblocks_k(P,s), s).
\]
Note that the parameter $k$ is the length by which the chain at state $s$ is extended (including through a chain-reorg) by $P$. The most common scenario will be extension by a single block for which we use will simply use $\MEV$ as shorthand henceforth. $\kMEV$ does not account for how difficult it is for $P$ to mine the $k$ consecutive blocks, but it is sufficient for our purpose to understand the value that can be extracted if a single miner could append multiple consecutive blocks. In Appendix~\ref{appendix:GMEV}, we define a weighted notion of miner-extractable-value that takes the probability of appending multiple blocks into account. We call this ``weighted MEV'' or $\WMEV$.

\begin{remark}[Local vs global maximization]
The astute reader may notice that our definitions (along with our concrete CFF instantiation in Section~\ref{sec:kmodel}) only considers the maximum value extractable in some \textit{given state} $s$. This can be considered analogous to finding a ``local maximum'' in the search space, leaving open the possibility that a non-optimal EV computation in the current state may lead to a higher combined EV when future states are also considered. 

As a simple example, consider a transaction $\tx$ that gives a specific miner $P$ a profit of $1~\ETH$ if it is mined when a contract $C$ has state $c_1$ and $10~\ETH$ when the contract has state $c_2$. Assume that the state change from $c_1$ to $c_2$ can only be caused (irreversibly) by a different player $P'$. Now, if $P$ mines a block when $C$ has state $c_1$, local MEV maximization would say that it should include $\tx$ within its block. But if $P'$ later causes the state change to $c_2$ in a new transaction, then $P$ would have made $9~\ETH$ more if it waited to include $\tx$.


While it is theoretically possible to define a ``global maximum'' for EV, computing it requires knowing the probability distribution of future transactions, i.e., how new transactions will be created and ordered within blocks (including by other players). In other words, it requires perfect knowledge of the strategy of all other players in the system, which is unrealistic.

We therefore focus in this work only on the maximum EV for a particular state. We emphasize however, that our definition is exact w.r.t. this local value.
\label{remark:locglob}
\end{remark}

\begin{remark}[MEV subsumes other attacks]
\label{remark:mevsubsume}
We highlight that our notion of MEV subsumes not only arbitrage but \textit{all attacks that can be carried out based on the current state of the system by a profit-seeking player}. Notably, this includes not only common strategies such as frontrunning, backrunning, and sandwich attacks~\cite{zhou2020high}, but also attacks with significant complexity observed in the wild, such as~\cite{pancakebunnyattack,arrayfinanceattack}. 

A common theme within these complex attacks in particular has been to use flash loans to borrow a significant amount of some token(s) and use this capital to extract profit by violating an implicit assumption in another contract (e.g., the valuation of a pool or token), before returning the loan. Such attacks can be explored from the current state without requiring additional state changes from other players, thereby allowing for our local computation of extractable value. We further note that since a miner is in a strictly more privileged position than any other permissionless player in the system, these strategies are exploitable by a miner.  Moreover, in any competitive race to extract these opportunities, the miner will ultimately have the option to capture the resulting revenue. This provides intuition for why MEV is more general than arbitrage or attacks.

We include a concrete example of such a flash-loan based attack within CFF in Section~\ref{sec:other-attacks}.

Since we focus on economic security, we consider only profit-seeking players and our definition of MEV therefore does not capture attacks that exploit a vulnerability but do not necessarily result in financial gain. Such attacks are considered traditional exploits, not economic ones.

\end{remark}


\subsection{Decentralized Finance Instruments}
\label{subsec:defi_instruments}



\mypara{DeFi instruments.}
We define DeFi instruments quite broadly, as smart contracts that interact with tokens in some way other than through transaction fees. We provide three concrete examples of DeFi instruments, which we use in running examples throughout the paper and as building blocks to discuss properties at higher levels of abstraction.

In particular, we specify here: (1) A simplified Uniswap contract; (2) A simplified Maker contract; and (3) A simple betting contract. We note that while we use simplified versions of the original contracts, they are still useful as didactic tools and for analyzing the core semantic properties underlying contract composition. {\em Note, however, that our instantiations of the contracts in the CFF (see Section~\ref{sec:kmodel}) include the missing details, i.e., are complete and usable for real-world data.} 

\mypara{Uniswap contract.}
\label{subsec:uniswap}
The Uniswap automated market maker contract~\cite{adams2019uniswap} allows a player to execute exchanges between two tokens (usually ETH and another token), according to a market-driven exchange rate. The contract assumes the role of the counterparty for such an exchange.  Uniswap uses an automated market maker formula, called the $x \times y = k$ formula or the \textit{constant product} formula. We discuss a simplified version here that does not deal with liquidity provisions, transaction fees, and rounding. Abstractly, for tokens $\mathbf{X}$ and $\mathbf{Y}$, the number of coins $x$ and $y$ for these tokens in the contract always satisfies the invariant $x \times y = k$, where $k$ is a constant. This equation can be used to determine the exchange rate between $\mathbf{X}$ and $\mathbf{Y}$. If $\Delta x$ coins of $\mathbf{X}$ are sold (to the contract), $\Delta y$ coins of $\mathbf{Y}$ will be received (by the user) so as to satisfy:
\[
x \times y = (x + \Delta x) \times (y - \Delta y).
\]
\fig~\ref{fig:uniswap_contract} \iffull\else(Appendix~\ref{sec:moreuni}) \fi shows the pseudocode for our simplified Uniswap contract $\uniswapcon^{(\mathbf{X}, \mathbf{Y})}$ for the tokens $\mathbf{X}$ and $\mathbf{Y}$. It contains a function $\texttt{exchange()}$ which allows a user to sell $\textrm{InAmount}$ tokens of $\textrm{InToken}$ to the contract in exchange for $\textrm{OutToken}$ tokens where $(\textrm{InToken}, \textrm{OutToken}) \in \{(\mathbf{X}, \mathbf{Y}), (\mathbf{Y}, \mathbf{X})\}$. The number of $\textrm{OutToken}$ tokens received by the user is given by the $x \times y = k$ market maker formula.

\iffull
\begin{figure}[!t]
\fpage{0.9}{
    \begin{center}
        \textbf{Contract $\uniswapcon^{(\mathbf{X}, \mathbf{Y})}$}
    \end{center}
    
     $\functioncode \texttt{exchange(\textrm{InToken, OutToken, InAmount}):}$ \\
    \myind \ifcode $\accbalance(\calleracc)[\textrm{InToken}] \geq \textrm{InAmount}$ \thencode \\
    \myind \myind $x = \accbalance(\uniswapcon)[\textrm{InToken}]$ \\
    \myind \myind $y = \accbalance(\uniswapcon) [\textrm{OutToken}]$\\
    \myind \myind $\textrm{OutAmount} = y - {xy}/{(x + \textrm{InAmount})}$\\
    \myind \myind $\accbalance(\calleracc)[\textrm{InToken}] \minuseq \textrm{InAmount}$\\
    \myind \myind $\accbalance(\calleracc)[\textrm{OutToken}] \pluseq \textrm{OutAmount}$\\
    \myind \myind $\accbalance(\uniswapcon)[\textrm{InToken}] \pluseq \textrm{InAmount}$\\
    \myind \myind $\accbalance(\uniswapcon)[\textrm{OutToken}] \minuseq \textrm{OutAmount}$
    
    \myind \elsecode Output $\bot$
    
}
\caption{Simplified abstract Uniswap contract}
\label{fig:uniswap_contract}
\end{figure}

\fi

\mypara{Maker contract.} 
\iffull
The Maker protocol allows users to generate and redeem the collateral-backed ``stablecoin" Dai through Collateralized Debt Positions (CDPs). Users can take out a loan in Dai by depositing the required amount of an approved cryptocurrency (e.g., \ETH) as collateral, and can pay back the loan in Dai to free up their collateral. If a user's collateral value relative to their debt falls below a certain threshold called the ``Liquidation Ratio" ($>1$), then their collateral is auctioned off to other users in order to close the debt position. Maker uses a set of external feeds as price oracles to determine the value of the collateral. A separate governance mechanism is used to determine parameters like the Liquidation Ratio, stability fees (interest charged for the loan), etc., and also to approve external price oracle feeds and valid collateral types. We consider here a simplified version of Maker's single-collateral CDP contract that does not model stability fees, or liquidation penalties. The contract $\makercon^{(\mathbf{X}, \mathbf{Y})}$ allows users to take out (or pay back) loans denominated in token $\mathbf{X}$ by depositing (or withdrawing) the appropriate collateral in token $\mathbf{Y}$, and allows for liquidation as soon as the debt-to-collateral ratio drops below the Liquidation Ratio. The contract is detailed in \fig~\ref{fig:maker_contract}.

It should be noted that the amount of collateral liquidated and received by the liquidator as well as the debt (in Dai) paid off by the liquidator in exchange for the collateral depends on the outcome of a 2-phase auction. If the auction is perfectly \emph{efficient}, the winning bidder pays off an equivalent amount of debt for receiving the offered collateral. On the other hand, when the auction is inefficient due to system congestion, collusion, transaction censoring, etc., the winning bidder can receive the entire collateral on offer without paying off an equivalent amount of debt. In our simplified Contract $\makercon^{(\mathbf{{X}, \mathbf{Y}})}$, we assume that liquidation is perfectly efficient.
\begin{figure}[!t]
\fpage{0.9}{
    \begin{center}
        \textbf{Contract $\makercon^{(\mathbf{{X}, \mathbf{Y}})}$}
    \end{center}
    
    $\liquidationratio = 1.5; \collateral = \{\}; \debt = \{\};$ \\
    
    $\functioncode \texttt{deposit\_collateral}(\qty):$\\
    \myind \ifcode $\accbalance(\calleracc)[Y] \geq \qty$ \thencode \\
    \myind \myind $\accbalance(\calleracc)[Y] \minuseq \qty$\\
    \myind \myind $\accbalance(\makeracc)[Y] \pluseq \qty$\\
    \myind \myind $ \collateral[\caller] \pluseq \qty$ \\
    
    $\functioncode \texttt{deposit\_loan}(\qty):$\\
    \myind \ifcode $\accbalance(\calleracc)[X] \geq \qty$ \andcode $\debt[\caller] \geq \qty$ \thencode \\
    \myind \myind $\accbalance(\calleracc)[X] \minuseq \qty$\\
    \myind \myind $ \debt[\caller] \minuseq \qty$ \\

    $\functioncode \texttt{withdraw\_collateral}(\qty):$\\
    \myind \ifcode $\collateral[\caller] \geq \qty$ \andcode  $\texttt{getprice}(Y, X) * (\collateral[\caller] - \qty) - \liquidationratio * \debt[\caller] \geq 0$\xspace\thencode \\
    \myind \myind $\accbalance(\calleracc)[Y] \pluseq \qty$\\
    \myind \myind $\accbalance(\makeracc)[Y] \minuseq \qty$\\
    \myind \myind $ \collateral[\caller] \minuseq \qty$ \\
    
    $\functioncode \texttt{withdraw\_loan}(\qty):$\\
    \myind \ifcode $\texttt{getprice}(Y, X) * \collateral[\caller] - \liquidationratio * (\debt[\caller] + \qty) \geq 0$ \thencode \\
    \myind \myind $\accbalance(\calleracc)[X] \pluseq \qty$\\
    \myind \myind $ \debt[\caller] \pluseq \qty$ \\

    $\functioncode \texttt{liquidate}(acc):$\\
    \myind \ifcode $\texttt{getprice}(Y, X) * \collateral[acc] - \liquidationratio * \debt[acc] < 0$ \thencode \\
    \myind \myind $\accbalance(\calleracc)[X] \minuseq \debt[acc]$\\
    \myind \myind $\accbalance(\calleracc)[Y] \pluseq \debt[acc] / \texttt{getprice}(Y, X)$\\
    \myind \myind $\accbalance(\makeracc)[Y] \minuseq \debt[acc] / \texttt{getprice}(Y, X)$\\
    \myind \myind $\debt(acc) = 0$\\
    \myind \myind $\collateral[acc] \minuseq \debt[acc] / \texttt{getprice}(Y, X)$\\
    
    $\functioncode \texttt{getprice}(Y, X):$\\
    \myind return $\frac{\accbalance(\uniswapcon)[X]}{\accbalance(\uniswapcon)[Y]}$ \\
}
\caption{Maker contract}
\label{fig:maker_contract}
\iffull \else \vspace{-2em} \fi
\end{figure}

\else 
We also model Maker, a popular DeFi protocol. The model is described in Appendix~\ref{subsec:maker}.
\fi



\mypara{Betting contract.}
To better understand composition failures, we introduce a simple betting contract and study its interaction with the previous contracts. Abstractly, the betting contract allows a user to place a bet against the contract on a future token exchange rate as determined by using Uniswap as a price oracle. By price oracle, we mean that the exchange rate between tokens as determined by the Uniswap contract is used to drive decisions in another contract. 

In \fig~\ref{fig:betting_contract}, we specify the contract $\betcontract^{\mathbf{X}}$ that takes bets on the relative future price of token $\textbf{X}$ to $\ETH$. Specifically, suppose that $\betcontract^{\mathbf{X}}$ is initialized with a deposit of 100 $\ETH$ tokens. A user Alice can now call $\texttt{bet()}$ and deposit 100 of her own $\ETH$ tokens to take a position against the contract. If at some point before the expiration time $t$, the Uniswap contract $\uniswapcon^{(\textbf{X},\ETH)}$ contains more $\ETH$ tokens than $\textbf{X}$ tokens, (i.e., the Uniswap contract values \textbf{X} more than $\ETH$), Alice can call $\texttt{getreward()}$ to claim 200 $\ETH$ from the contract, which includes her initial 100 $\ETH$ bet, along with her 100 $\ETH$ reward. Otherwise, Alice loses her initial bet.

For simplicity, our contract only contains a single bet, but it is straightforward to design similar contracts with more restrictions and/or functionalities (e.g., allowing another user to play the counterparty in the bet).



\begin{figure}

\fpage{0.9}{
    \begin{center}
        \textbf{Contract $\betcontract^{\mathbf{X}}$}
    \end{center}
    
    $\textsf{hasBet} = \false; \textsf{player} = \bot$ \\
    \texttt{// Contract also initialized with 100 ETH tokens when created.} \\

    $\functioncode \texttt{bet():}$\\
    \myind \ifcode $(\textsf{hasBet} = \false)$ \andcode $\accbalance(\account_{\textnormal{caller}})[\ETH] \geq 100$ \thencode \\
    \myind \myind $\accbalance(\account_{\textnormal{caller}})[\ETH] \minuseq 100$\\
    \myind \myind $\accbalance(\betcontract)[\ETH] \pluseq 100$ \\
    \myind \myind $\textsf{hasBet} = \true; \textsf{player} = \textnormal{caller}$ \\
    \myind \elsecode Output $\bot$\\
    
    $\functioncode \texttt{getreward():}$\\
    \myind \ifcode $(\textsf{hasBet} = \true)$ \andcode $\frac{\accbalance(\uniswapcon^{(\textbf{X},\ETH)})[\ETH]}{\accbalance(\uniswapcon^{(\textbf{X},\ETH)})[\textbf{X}]} > 1$ \andcode $(\textsf{player} = \textnormal{caller})$ \andcode (current time is at most $t$) \thencode \\
    \myind \myind $\accbalance(\account_{\textnormal{caller}})[\ETH] \pluseq 200$ \\
    \myind \myind $\accbalance(\betcontract)[\ETH] \minuseq 200$ \\
    \myind \elsecode Output $\bot$
}
\caption{Betting Contract $\betcontract$}
\iffull \else \vspace*{-2em} \fi
\label{fig:betting_contract}
\end{figure}


\section{DeFi Composability}
\label{sec:composability}



Smart contracts don't exist in isolation. A natural question, therefore, is when contracts ``compose securely.'' Abstractly, for a particular notion of security, does the security of a contract $C_1$ change when another contract $C_2$ is added to the system? In this paper, since our primary motivation is to analyze DeFi instruments, we focus on an economic notion of composable security. In particular, we look at how the extractable value of the system changes when new contracts are added to it. The \textit{economic} composability of an existing DeFi instrument $C_1$ w.r.t. $C_2$ now pertains to the added monetary value that can be extracted if $C_2$ is introduced into the system. That is, $C_1$ is composable w.r.t. $C_2$ if adding $C_2$ to the system does not give an adversary significantly higher extraction gains. For brevity, throughout this paper, we let {\em composability} refer to this specific notion, but note that it is orthogonal to previously considered notions (in, e.g.,~\cite{kosba2016-hawk,liao2019-ilc}).

Ideally, we want contracts to be ``robust'' enough to compose securely with all other contracts. Unfortunately, this may be too strong a notion in practice. We thus parameterize our definitions to allow  restricted or partial composability. Definition~\ref{def:single_defi_compose} defines the simplest notion of contract composability. 

\begin{definition}[Defi Composability]
Consider state $s$ and player $P$. A DeFi instrument $C'$ is $\varepsilon$-composable under $(P,s)$ if 
\[
        \MEV(P,s') \leq (1+\varepsilon)\;\MEV(P,s).
    \]
Here $s'$ is the state resulting from executing a transaction that adds the contract $C'$ to $s$ (no-op if $C'$ already exists). Although the composability of $C'$ pertains to all contracts in $\contracts(s)$, when looking at the specific interaction with a $C \in \contracts(s)$, we may also write that $C'$ is $\varepsilon$-composable with $(C,P,s)$. 
\label{def:single_defi_compose}
\end{definition}

In other words, allowing a player to interact with contract $C'$ in a limited capacity (using at most the tokens that the player controls in $s$) does not significantly increase the profit the player can extract form the system. Note that Definition~\ref{def:single_defi_compose} can easily be extended to consider several states and or players.



\subsection{Characteristics of Contract Composition}
We find that DeFi instruments that are secure under composition according to Definition~\ref{def:single_defi_compose} are surprisingly uncommon, especially when two instruments depend on each other (e.g., one contract using the other as a price oracle). Intuitively, manipulating one contract can change the execution path of the other contract. In this section, we analyze the composition among the contracts ($\uniswapcon$, $\betcontract$, and $\makercon$) introduced in Section~\ref{subsec:defi_instruments} to highlight interesting characteristics that can arise from smart contract composition. Note that for this simplified, didactic analysis, we do not make use of our CFF tool. We summarize our observed characteristics below.

\begin{characteristic}
Composability is state dependent---contracts may be $\varepsilon$-composable in state $s$ but not in another state $s'$.
\end{characteristic}

\begin{characteristic}
Composability depends on the actions allowed for a player. For instance, contracts may be composable if only transaction reordering is allowed but not if the creation of new transactions is allowed as well.
\end{characteristic}

\begin{characteristic}
A contract may not be composable with another instance of itself.
\end{characteristic}

\begin{characteristic}
It is often possible to introduce {\em adversarial} contracts that break composability with minimal resources. Thus it is important to consider composability not just of existing contracts, but also over such adversarial contracts.
\end{characteristic}

To provide intuition for these properties, we will analyze the following contract compositions. Section~\ref{subsec:price_oracle} considers the use of $\uniswapcon$ as a price oracle for either $\betcontract$ or $\makercon$. \iffull Section~\ref{subsec:multiple_amm} \else Appendix~\ref{subsec:multiple_amm}\fi analyzes the composition between multiple independent instances of $\uniswapcon$. \iffull Section~~\ref{subsec:mev_bribery}\else Appendix~\ref{subsec:mev_bribery}\fi introduces a new bribery contract that can be used to inject non-composability into the system. 


\subsection{Uniswap as a Price Oracle}
\label{subsec:price_oracle}

\begin{example}[$\uniswapcon$ as a price oracle for $\betcontract$]
Consider a simplified Uniswap contract ($\uniswapcon$) that exchanges the tokens $\BBT$ and $\ETH$, and a betting contract ($\betcontract$) that uses it as a price oracle.

In particular, consider a system state $s$ such that $\contracts(s) = \{\uniswapcon\}$ (or alternatively $\contracts(s)$ contains other contracts that do not affect the composability). Suppose that in state $s$, $\uniswapcon$ contains $b$ $\BBT$ tokens and $e$ $\ETH$ tokens such that $b > e$. To denote the Uniswap transactions contained in the mempool in state $s$: 
\begin{itemize}
    \item Let $\mathcal{T}_{B \rightarrow E}$ be the set of transactions that sell $\BBT$ tokens to the contract in exchange for $\ETH$ tokens. Suppose the total number of $\BBT$ tokens transacted is $b'$.
    \item Let $\mathcal{T}_{E \rightarrow B}$ be the set of transactions that sell $\ETH$ to the contract in exchange for $\BBT$ tokens. Suppose that the total number of $\ETH$ transacted is $e'$.
\end{itemize}
For a player $P$, let $p_e$ and $p_b$ be the number of $\ETH$ and $\BBT$ tokens held by $P$ in the state $s$ that are not within pending transactions in the mempool. Note that $P$ can use transactions from other accounts within the mempool as well as any transactions it can create with its own capital to create a block. Note that even if $P$ does not have the hash power to mine blocks, it can pay some other miner to order transactions according to its preference. Let $s'$ be the state resulting from adding $\betcontract$ to state $s$.
\label{example:uniswap-pricebet}
\end{example}

\mypara{Composability is state dependent.}
It is easy to see that contracts that are independent of each other and provide orthogonal functionalities should compose securely in all states.  In most real-world cases, however, we want to analyze the composability of contracts that are not independent and may in fact depend on each other's state. In such situations, whether two contracts compose securely will almost always depend on the characteristics of the current system state.

We use Example~\ref{example:uniswap-pricebet} to provide intuition to this observation. Specifically, we show that $\uniswapcon$ and $\betcontract$ are composable in states with a small number of available tokens, while in other states, an adversary can extract more MEV from the composition. Suppose that we define the number of \textit{liquid tokens} in the Uniswap contract as follows: For player $P$ and state $s$, we say that there are $l_b = l_b(P,s) = b' + p_b$ liquid $\BBT$ tokens and $l_e = l_e(P,s) = e' + p_e$ liquid $\ETH$ tokens.
We will now show how composability can be affected by the number of liquid tokens in the current state.

\begin{enumerate}[label=\alph*), wide]
    \item \textit{Composability in states with a small number of liquid tokens.} When $l_e \leq b - e$, i.e., the number of liquid tokens is sufficiently small, $\uniswapcon$ and $\betcontract$ do in fact compose securely. This is because regardless of what transactions $P$ creates or how it orders existing transactions in the transaction pool, at no point in the execution of a created block can the number of $\ETH$ tokens in $\uniswapcon$ exceed the number of $\BBT$ tokens in it. 
    In other words, $P$ cannot maliciously create a short term fluctuation in the exchange rate in order to claim a reward from $\betcontract$. Note that while $P$ can still cause the exchange rate to be manipulated even if it cannot cause the number of $\ETH$ tokens to exceed the number of $\BBT$ tokens, since we are focusing only on composability with $\betcontract$ here specifically, $P$ will not be able to claim the reward from $\betcontract$.
    
    Consequently, any value that $P$ can extract in state $s'$ (obtained by adding  $\betcontract$ to state $s$) can also be extracted in state $s$. Equivalently, $\MEV(P,s') = \MEV(P,s)$. We conclude that $\uniswapcon$ is $0$-composable under $(\betcontract, P, s)$.

    \item \textit{Non-composability in other states.} Suppose now that our low liquidity assumption was no longer valid. In particular, we will consider states $s$ such that $e' > b - e$, and $p_e \geq 100$. At least 100 $\ETH$ is necessary in our example to actually take a bet against the betting contract.
    To extract more value in state $s$, a malicious miner $P$ can proceed as follows:
    \begin{enumerate}[label=\arabic*),leftmargin=\parindent, labelindent=\parindent]
        \item Insert a transaction that takes a bet against the contract $\betcontract$ by depositing 100 $\ETH$.
        \item Order all transactions in the set $\mathcal{T}_{E \rightarrow B}$ . This raises the amount of $\ETH$ in  $\uniswapcon$ temporarily.
        \item Insert a transaction (a call to \texttt{getreward()}) to claim the reward of 100 $\ETH$ (in addition to its original bet) from $\betcontract$ due to the short term price fluctuation in $\uniswapcon$.
        \item Order the transactions in $\mathcal{T}_{B \rightarrow E}$  to buy $\ETH$ from $\uniswapcon$.
    \end{enumerate}

Abstractly, by ordering all transactions that sell $\ETH$ to $\uniswapcon$ first, $P$ can create a short-term volatility in the exchange rate between $\ETH$ and $\BBT$, allowing $P$ to claim the reward from $\betcontract$. When the block created by $P$ executes, since all transactions that add $\ETH$ to $\uniswapcon$ are ordered first, there will be more $\ETH$ tokens than $\BBT$ tokens by the time the $P$'s transaction to claim the reward from $\betcontract$ executes. This sudden change in the amount of $\ETH$ is only temporary as the remaining transactions in the block will reduce the number of $\ETH$ tokens. Note that this reordering attack is still possible in the case that $b' \approx e'$ and the natural or ``fair'' transaction order would not cause such a large change in the exchange rate during normal execution.  
Yet, the malicious miner $P$ was able to profit simply by reordering user transactions. 
\end{enumerate}

\mypara{Composability depends on the allowed actions.}
In the context of Example~\ref{example:uniswap-pricebet}, if $P$ cannot insert its own transactions for $\uniswapcon$, then composability holds even if $p_e + e' - 100 > b - e > e'$ and $p_e \geq 100$, since $P$ cannot create a large enough price fluctuation simply from the transactions in the mempool. However, if $P$ has the ability to insert its own transactions, it can use the previously mentioned procedure to extract the reward from $\betcontract$. $P$ can also insert its transactions before and after user transactions to take advantage of the short term slippage in the Uniswap price. This strategy resembles the sandwiching attack described in~\cite{zhou2020high}, which combines frontrunning and backrunning. It also allows $P$ to capitalize on the price differential between limit orders and market orders.

\iffull 
\mypara{Uniswap as a price oracle for Maker.}
Similar problems would arise if Uniswap is used as a price oracle in the Maker contract. By reordering Uniswap transactions, and thereby manipulating the exchange rate, a miner can cause the value of a user's collateral to fall below the acceptable threshold, and trigger a liquidation event. Furthermore, the miner can buy the user's collateral tokens in the liquidation event, and later sell them for a profit when the exchange price returns to normal.
\fi

\iffull
\subsection{Composition of multiple AMMs}
\label{subsec:multiple_amm}
Perhaps surprisingly, we find that even multiple contracts deployed with the same code need not be composable with each other. An interesting example of this non-composability is seen when two automated market makers (AMM) contracts co-exist in a system. Example~\ref{example:multiple_amm} highlights this observation.

\begin{example}
Consider state $s$ containing two instances, $\uniswapcon$ and $\uniswapcon^{*}$, of the Uniswap contract that exchange between the same two tokens ($\BBT$ and $\ETH$). 
Let $b, e$ be the number of $\BBT$ and $\ETH$ tokens respectively in $\uniswapcon$, and let $b^*, e^*$ be the number of $\BBT$ and $\ETH$ tokens respectively in $\uniswapcon^*$.
\label{example:multiple_amm}
\end{example}

\begin{restatable}{lemma}{multipleamm}
If $be^* \neq b^*e$, then there exists a $\delta > 0$ such that for any $0 < \alpha < \delta$, a miner with at least $\alpha$ $\ETH$ (equiv. $\BBT$) tokens can achieve an end balance of more than $\alpha$ $\ETH$ (equiv. $\BBT$) tokens
by only interacting with $\uniswapcon$ and $\uniswapcon^*$.
\end{restatable}
\begin{proof}
We prove for $\ETH$ tokens but note that the proof is exactly the same for $\BBT$ tokens. Let $U = \{\uniswapcon, \uniswapcon^*\}$. Consider the following sequence of transactions: (1) Deposit $\ETH$ in contract $A \in U$ to retrieve tokens of $\BBT$; (2) Deposit the $\BBT$ tokens in $A' \in U \setminus A$ to get tokens of $\ETH$. We will show that when $be^* \neq b^*e$, there exists a $\delta > 0$ such that depositing $\alpha$ ($0 < \alpha < \delta$) tokens in step (1) results in more than $\alpha$ tokens in step (2).

First, suppose that $\alpha_0$ $\ETH$ tokens are deposited in $\uniswapcon$ in the first step. This results in $\frac{b \alpha_0}{e + \alpha_0}$ $\BBT$ tokens, which when deposited in $\uniswapcon^*$ gives back $\frac{be^*\alpha_0}{b^*e+b^*\alpha_0 + b\alpha_0}$ $\ETH$ tokens. Similarly, if $\alpha_0$ $\ETH$ tokens were first deposited in $\uniswapcon'$, then the user would end up with $\frac{b^*e\alpha_0}{be^* + b\alpha_0 + b^*\alpha_0}$ $\ETH$ tokens. Now, we consider the following cases:

\textbf{Case (1)} $be^* - b^*e > 0$.
Let $\delta = \frac{be^* - b^*e}{b + b^*}$. Therefore, $b^*e + b \alpha + b^* \alpha < be^*$ which gives $\alpha < \frac{be^*\alpha}{b^*e + b \alpha + b^* \alpha}$. In other words, depositing first in $\uniswapcon$ and then in $\uniswapcon^*$ yields more $\ETH$ tokens than the initial deposit.

\textbf{Case (2)} $be^* - b^*e < 0$.
This is analogous to the first case.
Let $\delta = \frac{b^*e - be^*}{b + b^*}$. Therefore, $be^* + b \alpha + b^* \alpha < b^*e$ which gives $\alpha < \frac{b^*e\alpha}{be^* + b \alpha + b^* \alpha}$. In other words, depositing first in $\uniswapcon^*$ and then in $\uniswapcon$ yields more $\ETH$ than the initial deposit.
\end{proof}

\subsection{MEV Bribery Contracts}
\label{subsec:mev_bribery}
New contracts can be introduced into the system specifically with the goal of breaking composability. One such example is that of \textit{bribery contracts}. The existence of MEV in a system can give rise to new bribery-based incentives for miners to choose the final transaction ordering. For instance, a user could bribe a miner to give her transactions preferential treatment (e.g., a better exchange rate for Uniswap transactions). Such bribes can be carried out securely through bribery contracts. Consider the following simple example. 

\begin{example}
A user $U$ and a miner $P$ enter into a bribery smart contract with a payout as follows: $P$ submits two valid transaction orderings, $O_1$ and $O_2$, such that $O_1$ is preferred by $U$; if $O_1$ is the finalized order, $P$ receives a payout proportional to the difference to the user $U$ in value of $O_1$ and $O_2$. 
\end{example}

Intuitively, $U$ is ``bribing'' the miner to provide $U$ with a more profitable transaction ordering. To maximize its profit, a miner may potentially enter into multiple such bribery contracts with other users, and pick the best one to complete. Bribery contracts could also pose a threat to the long term stability of the system; given enough incentive, it could be worthwhile to mine a consensus block on a stale chain, thereby attempting to rewrite blockchain history. This is similar to time-bandit attacks, which as observed in~\cite{daian2020flashboys} can be highly detrimental for current blockchain consensus protocols.

\fi

\subsection{Remarks on Composability}

\iffull \else
We provide some additional exploration of composability and its relationship to bribery and oracles in Appendix~\ref{sec:morecomposability}.
\fi
We end with some remarks on our composition examples.

\mypara{Takeaways for smart contract developers.}
Unfortunately, as our composition examples show, the security of a DeFi smart contract may not always depend solely on the contract's code; design flaws in other contracts---even those deployed much later---may cause composability failures. This is problematic for contract developers since it implies that security of their contracts may in fact be out of their hands.

\mypara{Remark on capital requirements.}
Several of our DeFi composability attacks in this section require the miner to possess some initial capital to carry out malicious transaction reorderings and extract MEV. 
Despite this, we note that in the real world, capital requirements will rarely be barriers to exploiting the system, even for smaller players, particularly due to the availability of flash loans.Flash loans are essentially risk-free loans that can be offered any time arbitrage or other profitable system behavior can be executed atomically, which is often the case. Flash loans also do not compose with contracts that were designed without flash loans; the attacks in~\cite{qin2020attacking} are an example of this. Consequently, adding flash loans to any of our non-composability examples will only exacerbate the impact of malicious transaction reordering.

\section{Clockwork Exploration in K}
\label{sec:kmodel}
{
\color{red}

\begin{figure}[!t]
    \centering
    \includegraphics[scale=0.85]{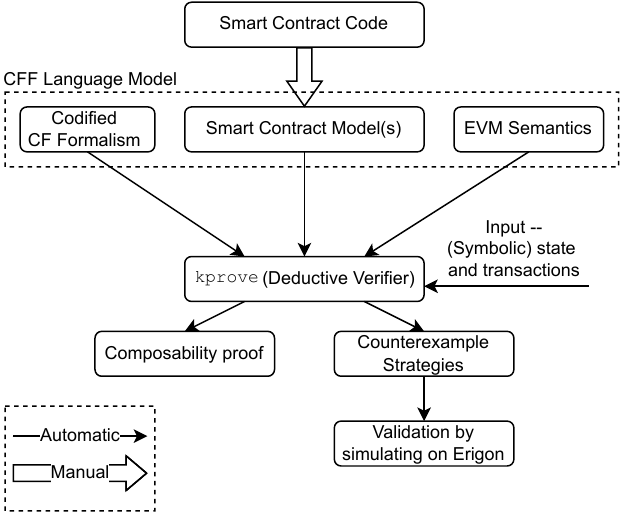}
    \caption{CFF architecture}
    \label{fig:cff_arch}
    \iffull \else \vspace*{-2em} \fi
\end{figure}
}
\iffull
Equipped with our formalism for reasoning about the security of DeFi instruments, we now discuss how best to apply it to real-world contracts.
\fi
To establish a formal methodology for DeFi security, we instantiate our Clockwork Finance Framework (CFF) in the K framework for mechanized proofs. 
\iffull
Appendix~\ref{subsec:considerations} elaborates on why we chose K. 
\else
We include a discussion in the full version~\cite{fullversion} on why we chose K. 
\fi

\iffull 
We first describe challenges with formal verification and how we overcome them for CFF (Section~\ref{subsec:challenges}). We describe the design and implementation of CFF, with an emphasis on the soundness and completeness properties in Section~\ref{subsec:design_impl}. We then discuss how our CFF executable models are obtained and their properties in Section~\ref{subsec:equivalence}. Finally, we use the Uniswap contract (\fig~\ref{fig:uniswap_contract}) as an example to describe our CFF executable models (Section~\ref{subsec:usability}).
\fi



\subsection{Scaling Formal Verification for CFF}
\label{subsec:challenges}
Unfortunately, simply applying formal verification tools out-of-the-box to our models turns out to be impractical. To understand why, we need to step back and consider the number of paths from the start of model execution to termination of execution that must be explored by any formal verification tool, in an attempt to exhaustively prove a specific property holds in all possible executions. While general sound formal verification techniques are known to be undecidable, in practice they usually suffice for typical programs, where execution semantics are primarily linear. Branching conditions (e.g., control-flow branches) generally cause an increase in the number of paths to explore. Here, the number of paths that must be explored could be exponential in the number of branches in the program. 


However, in our setting, miners can choose \emph{any} ordering of transactions (others' transactions plus their inserted transactions) when creating a block. This means that the number of unique paths needed to fully explore the search space is $O(t!)$ where $t$ is the number of transactions to which we apply our CFF. This is asymptotically and concretely more expensive than usual program verification proofs, and consequently impractical for even a modest number of transactions. One existing parallel in the literature is to semantics of concurrency (see e.g.,~\cite{huang2014maximal}), in which many possible interleavings must be reasoned about. Nonetheless, most such tools either work with a small concurrency parameter, or do not attempt to exhaustively analyze the full state space of interleavings. They attempt only to find plausible bugs based on observed behavior.

\mypara{Search-space reduction.}
To make formal verification practical, we must first reduce the search space to a tractable set of paths. We found that reasoning about all possible transaction orders in the formal model directly results in a large amount of repeated computation as equivalent states are explored 
(e.g., by re-ordering non-dependent transactions). 

    Therefore, we apply the following optimizations (both general and DeFi instrument specific) to our analysis to reduce the number of paths by excluding semantically equivalent orderings. First, transactions carry a \textit{per user} serialization number (``nonce'') such that transactions that are mined out of order are considered invalid. Thus, we consider orderings equivalent if for each non-miner player, the longest consecutive (by nonce) subsequence of transactions is the same (since transactions not belonging to these subsequences are invalid). Second, transactions that interact with different contracts (such as swaps on different Uniswap pairs) are independent of each other. They produce equivalent orderings if reordered relative to one another. Third, we allow for models to incorporate  application-specific optimizations. We do so, for example, for our AMM models. The constant-product AMM function is provably path independent~\cite{buterin2017path}. For example, if the miner makes multiple sequential trades selling an asset, exploring their reorderings will have no effect. This optimization cuts the work required by our tool by orders of magnitude, and allows CFF to explore problem instances with larger number of transactions. Note that the above optimizations\footnote{We encode our optimizations in the \texttt{run\_uniswapv2\_experiments} \& \iffull the \fi \texttt{run\_mcd\_experiments} files provided in our Github repository.} are all sound. While we would ideally like to avoid application-specific optimizations even if sound, and our tool does support this, we found that they substantially improved performance. Similar optimizations will likely be helpful for any MEV analysis.

\subsection{Design and Implementation}
\label{subsec:design_impl}
    \fig~\ref{fig:cff_arch} shows the CFF architecture.
    The core of CFF is the language model whose syntax and semantics are fed to the K framework to automatically generate the deductive verifier \texttt{kprove} along with other tools for parsing, compiling, and symbolic execution of transactions. Note that because of gas limits on the size of a block and computation done in a transaction, the semantics of our language model are decidable. Due to~\cite{rosuverifier}, this implies that the deductive verifier we obtain is sound and complete for any reachability property of our language model. Since we model the problem of economic security as a reachability problem (of a state with certain MEV), CFF is attack exhaustive for the transactions and contracts it is given. Any sources of unsoudness in our verification come from our language model, which we now describe.
    
    The first component of our language model defines the specific parameters for the MEV computation as per in the CF model (Section~\ref{sec:model}). It starts with defining a transaction type, block type, and player types. A player of type ``miner'' can produce a block by deciding the order of the mempool transactions and any inserted new transactions. Note that the miner cannot manipulate others' transaction \textit{contents}, as transactions are digitally signed by their creators. While our formalism from Section~\ref{sec:model} allows for arbitrary transaction insertions (including inserting transactions that create new contracts!), our implementation, for tractability, only handles user-specified templates of inserted transactions. These are \textit{template transactions} because their calldata is allowed to have symbolic parameters rather than concrete values. The lack of arbitrary transaction insertions in our implementation is one source of unsoundness when CFF proves upper bounds on MEV as a measure of economic security. Fortunately, this is not a theoretical limitation since limits on block sizes in Ethereum and other blockchains also constrain the number and type of permissible insertions. (e.g., a transaction cannot exceed the block size). Moreover, arbitrary transaction insertions are observed only rarely in the wild, and incur high gas fees. Barring transaction insertions that create a contract, given enough computing resources, CFF can be extended to reason about all types of insertions by enumerating all possible interactions with the given contracts.
    
    The second component of our language model defines the semantics of the smart contract code and specific smart contract models. The K Framework has built-in semantics of basic arithmatic and logical operations. We enrich it with definitions of currency transfers and smart contract storage. These limited semantics are sufficient to express our smart contract models, and make the verification much faster than incorporating full EVM semantics. We then manually translate the smart contract code into CFF models written in K; we give details in Section~\ref{subsec:equivalence}. This needs to be done \textit{only once} for each contract. Note that our limited semantics of EVM and the way we obtain our CFF models mean that any successful trace obtained in the actual smart contract can be obtained in our CFF models (but not vice-versa). We elaborate on this in Section~\ref{subsec:equivalence}. As a result, the proofs of economic security found by CFF on the smart contract models for the given transactions also hold for the actual smart contracts (i.e., there are no false positives introduced here). However, this over-approximation introduces false negatives, i.e., the counterexample strategies (sequence of transaction) found by \texttt{kprove} may not all be valid on the actual smart contracts. To validate potential counterexample strategies, CFF simulates the sequence of transactions in these strategies on an archive node at the appropriate block height. This validation step is fully automatic and takes on average 39 milliseconds per counterexample with a standard deviation of 22 milliseconds. 
    
    We have contributed our implementation for simulating transactions at a given block height into the latest public release of the Erigon (popular Ethereum client) software and is now accessible via the \texttt{eth\_callBundle} JSON-RPC API.

    The gap between our smart contract models and the actual corresponding smart contracts can be closed by substituting the second component of our language model with KEVM~\cite{hildenbrandt2018kevm}. There is a tradeoff, however: the performance of CFF would degrade with use of KEVM. We leave exploration of KEVM integration to future work. We also believe there is room for a wide range of hybrid approaches, including randomized testing / fuzzing, symbolic execution, concolic testing~\cite{yun2018qsym}, and machine learning, to attempt to learn and optimize for this state transition model.

{
\subsection{Equivalence and Over-Approximation in CFF models}
\label{subsec:equivalence} 

\begin{figure*}
    \centering

\begin{lstlisting}
Status: SUCCESS
Returns: msg.value * 997 * token_reserve / ((self.balance - msg.value) * 1000 + msg.value * 997)
Path condition: deadline >= block.timestamp /\ eth_sold > 0 /\ min_tokens > 0 /\ not(#status(130) == 0) /\ self.balance - msg.value > 0 /\ token_reserve > 0 /\ (msg.value *Word 997) /Int msg.value == 997 /\ (input_amount_with_fee *Word output_reserve) /Int input_amount_with_fee == output_reserve /\ (input_reserve *Word 1000) /Int input_reserve == 1000 /\ not((input_reserve * 1000) + input_amount_with_fee < (input_reserve * 1000)) /\ not(tokens_bought < min_tokens) /\ not(#status(133) == 0) /\ not(#transferReturn(133) == 0)
\end{lstlisting}
\vspace{-10mm}

\begin{lstlisting}
Status: REVERT
Path condition: not(deadline >= block.timestamp and eth_sold > 0 and min_tokens > 0)
\end{lstlisting}
\vspace{-6mm}

\hrule

\begin{lstlisting}
Address in TokenOut gets (997 *Int TradeAmount *Int USwapBalanceOut) /Int (1000 *Int USwapBalanceIn +Int 997 *Int TradeAmount)
\end{lstlisting}
\vspace{-6mm}

\caption{Two example paths from Uniswap EVM contract verification through symbolic execution (above line, prior work~\cite{uniswapverification}), and corresponding CFF model return value formula (below line, \texttt{uniswap.k}).}
    \label{fig:unipaths}
    \iffull \else \vspace{-2em} \fi
\end{figure*}

We now discuss a general approach we used for creating our models. \textbf{This is not the only way to create CFF models}, but is the most formal possible approach, allowing for a clear equivalence between the EVM executing on-chain and the CFF model. The approach proceeds in three steps:

\begin{enumerate}[noitemsep]
    \item \textbf{Path decomposition/verification (before CFF):} Perform a path decomposition of the target smart contract, a standard technique required for formal verification of smart contracts in KEVM~\cite{hildenbrandt2018kevm} (outside of CFF). For the highest possible assurance, developing a fully validated model requires some developer effort beyond developing the EVM code, but minimal effort beyond developing a formal proof. Developing unvalidated models is possible, but in our development of CFF we have instead started with a formal proof about the target EVM code (see~\cite{uniswapverification}) and built a CFF model from there.
    \item \textbf{Pruning/selection and refinement:} Select all relevant paths in (1), prune reverting or non-MEV-relevant paths (e.g., utility functions), and import these remaining paths into a CFF model. This process can mainly be automated from (1), but some minimal developer judgment on which paths to include can improve analysis speed. 
    \item \textbf{Argument of equivalence:} If any changes to the obtained path formulas are desired, e.g., variable renaming for readability, argue equivalence of the CFF model in (2) to the path decomposition/formal EVM proof in (1) (see our example code for Uniswap equivalence).
\end{enumerate}

We expand on each on these three steps below.

\mypara{(1) Path decomposition.} The first step is simply performing a standard complete symbolic exploration of the EVM bytecode of the smart contract. This is a general pattern of smart contract development that is not specific to our work. To prove a contract correct in the K framework, K executes the EVM code against the KEVM semantics~\cite{hildenbrandt2018kevm} on fully symbolic input and EVM state, and decomposes all possible return values of the contract into a mathematical formula over all possible inputs. This involves many possible paths, which represent symbolic branches through the EVM contract code. A contract is said to be verified in K if desired security properties hold as invariants on every such path. A formal specification of a contract's behavior in K is equivalent to a specification of its behavior on each possible path.

This path decomposition step is not mandatory (one can simply directly give a mathematical specification as on the bottom of \fig~\ref{fig:unipaths} without decomposing EVM code), but it leads to high assurance models by construction, and requires little developer effort beyond a formal proof (which has independent value), so it is the technique we choose to describe.

This approach is standard for verifying high-assurance smart contracts. An ideal case study is provided specifically for Uniswap in a report commissioned by Uniswap to demonstrate the security of their contracts, described in~\cite{uniswapverification}. We directly use the results published for the Uniswap EVM contract by Runtime Verification Inc. of the process above to generate our CFF model of Uniswap. We execute their proofs of correctness for Uniswap to extract all paths in the EVM code. One such example path is shown in the upper box of \fig~\ref{fig:unipaths}, for the tokenToEthInput function, which swaps a token for ETH.

This generated path states that, if the listed path condition (Line 3) is met across input and world state (where the variable names have been manually labeled in some cases by the author of the formal proof, in this case Runtime Verification, Inc.), the return value of the EVM call (Line 2) will be successful and will output the formula listed. This formula contains variables that can be sourced from the input or world state.

The box just above the horizontal line in \fig~\ref{fig:unipaths} is another path in which EVM execution reverts when the input and world state meet different conditions.

\mypara{(2) Pruning/selection and refinement.} In our CFF model, we include a simplified variant of the top path, shown below the line in \fig~\ref{fig:unipaths}. We do not include the reverting bottom path, and can simplify the resulting path conditions (our model has no concept of e.g. deadlines).

By choosing to omit all reverting paths, we are able to study the properties of interactions between the compositions of non-reverting paths without reasoning about the complex branching and path conditions that may lead to these reverts, simplifying our underlying queries to K (the size of the Z3~\cite{z3solver} formula \texttt{kprove} queries on the backend is proportional to the complexity of the models~\cite{rosuverifier}).

Omitting reverts will never reduce the amount of MEV found by our search. The only consequence will be that some attack we explore \emph{would revert} in an actual execution, but will not in our analysis. This can only add, not remove, MEV to each execution. We allow for initial discovery of such executions through our automated tool, and filter them out through our automated validation described in~\ref{subsec:design_impl}.

\mypara{(3) Argument of equivalence.} The final step is to argue that each path in our CFF model is equivalent to a successful path generated by contract verification. There are two possibilities. One can manually algebraically inspect the formulas, reasoning about equivalence on-paper. There is a very direct argument in this case that the formulas are structurally the same by inspection, modulo variable renaming.

For automatic equivalence, one can turn to unification, a standard technique for creating a map of variable renamings in syntactically equivalent formulas, to create a substitution of variable names.  This can be automated to verify a large number of paths against automatically performed path decomposition. We provide an example argument using unification~\cite{unification} in the 
\texttt{cff\_model\_equivalence} directory. This example shows that our Uniswap CFF model is equivalent to the deconstructed paths from the Uniswap EVM code listed above it (arguing that the bottom and top of \fig~\ref{fig:unipaths} are equivalent). 

\vspace{2mm}

Using the above three-step approach, as we have demonstrated for Uniswap, yields several convenient properties of the resulting CFF models, which hold for all models we provide:

\mypara{Over-approximation.} Following this technique for model construction, any resulting model is an over-approximation of the EVM bytecode: it models exactly \emph{all non-reverting paths} on which the underlying contract successfully executes a transaction, and avoids modeling code paths in the contract bytecode or EVM-related semantic rules/details that do not affect relevant state or balances.

Such a model will over-approximate attacks, yielding some attacks that do not actually work on-chain because they may trigger an unmodeled reverting path (which we call \emph{false positives}). Because weeding out false positives is cheap and easily parallelizable, while reasoning about attacks is expensive and scales with underlying code complexity, the less literal approach of simplifying our model and filtering out reverting paths as needed allows us to explore a wider space of attacks than use of an exact but more complex model. 

Our techniques do not generate \emph{false negatives}, or non-reverting paths that could have occurred in practice but are not explorable by our search. This is because we maintain all non-reverting paths in our models, and strictly relax the relevant path conditions, as we show by example for Uniswap.

We say that under this relaxation---which allows for false positives but not false negatives---our models are \emph{over-approximations} of the underlying contracts.

\mypara{Development overhead.} Note that constructing the models according to the three-step strategy we've described requires virtually no developer effort/overhead for a developer who has already created a formal proof of contract correctness. Because formal verification is a popular technique for high-assurance contracts, in many cases, robust CFF models can be extracted from existing formal models with minimal additional developer effort. If developers do not want to formally verify their contracts, their CFF models must be coded manually and may prove less secure, as they will need to manually reason about or concretely validate the models' correctness against an EVM deployment (Section~\ref{subsec:expvalidation}). Note that this practice is still supported by our framework: we allow for reasoning about models that are not created using our three-step approach, or may be different than the EVM code they represent, as this may be useful for creating new contracts, perhaps before EVM code is even developed. Our intent is here instead to showcase the possibility and process for developing high-assurance, useful models such as our Uniswap model.

\mypara{Constant model overhead.} If models are developed using the above technique of symbolic path decomposition, we argue that our model size has a constant overhead compared to the corresponding smart contracts.  
In our work, the model used for verification is only the set of paths we deem relevant. Because we strictly remove paths and conditions from the verified EVM to create an over-approximation, our models are by definition smaller in both number and complexity of semantic rules than a complete contract model (the two relevant scaling metrics for formal language models). While the exact number of paths removed depends on the target contract, this puts our approach in contrast to approaches such as~\cite{zhou2021just}, which require, e.g., a path definition for each token pair, and thus scales poorly in size compared to the EVM contract itself. 


}

\subsection{CFF Uniswap Model}
\label{subsec:usability}
\begin{figure*}
\lstset{language=XML}
\begin{lstlisting}
<k> exec(Address:ETHAddress in TokenIn:ETHAddress swaps TradeAmount:Int input for TokenOut:ETHAddress)  => 
    AmountToSend = (TradeAmount *Int USwapBalanceOut) /Int (USwapBalanceIn +Int TradeAmount);
    Address in TokenIn gets 0 -Int TradeAmount;
    Address in TokenOut gets USwapBalanceOut -Int var(AmountToSend);
    Uniswap in TokenIn gets TradeAmount;
    Uniswap in TokenOut gets 0 -Int var(AmountToSend);
    ...
</k> 
<S> ... (Uniswap in TokenOut) |-> USwapBalanceOut (Uniswap in TokenIn) |-> USwapBalanceIn ... </S> 
<B> ... .List => ListItem(Address in TokenIn swaps TradeAmount input for TokenOut) </B> \end{lstlisting}
\vspace{-6mm}
\caption{Simplified Uniswap contract implemented in CFF. Ellipses match the rest of the program state in each cell.}
\label{fig:uniswap_k_contract}
\iffull \else \vspace{-1.5em} \fi
\end{figure*}

\fig~\ref{fig:uniswap_k_contract} shows an implementation (in K) of a snippet of our abstract Uniswap contract from \fig~\ref{fig:uniswap_contract}, the same contract we developed above using path decomposition. This refines our presented abstract contract and formalism  and transforms it into a computer-readable executable model, capable of being symbolically and concretely reasoned about by the symbolic execution engine and deductive verifier bundled with K.
\iffull 

A few key differences exist between our abstract contract and executable CFF models.  The first is that our executable CFF models contain an XML-like configuration consisting of \emph{cells}, or mathematical objects in the K Framework.  The \kcell, \scell, and \bcell cells of our executable model are featured in \fig~\ref{fig:uniswap_k_contract}.  Recall that our model represents a state machine executing Uniswap transactions.  The \kcell cell  specifies the transactions left to execute in the model and not yet included in a block, and can be viewed similarly to a program tape in a Turing-style execution machine. Note that execution of these transactions by CFF takes different paths corresponding to different orderings (including the original order in \kcell cell) and censoring combinations of these transactions.  The \scell cell represents the space of state mapping $S$ in CFF (Section~\ref{sec:model}), and stores a mapping of addresses to balances (state entries).  The \bcell cell represents the prefix of the block that has been constructed thus far by CFF. The model is consistent with our formalism by maintaining the invariant : $\scell = \blockaction(\bcell)(s_0)$ where $s_0$ is the initial state. When no instructions are left to execute by CFF (empty \kcell cell), the \bcell cell will represent a valid block. The final state and the contents of the valid block potentially vary for different execution paths. 

Another key difference is that our abstract contract has imperative semantics while K is fundamentally rewrite-based~\cite{chen-rosu-2018-isola} using ``A $=>$ B" as a special operator meaning ``A rewrites to B". Lines 1-6 in~\fig~\ref{fig:uniswap_k_contract} correspond to one of the rewrite operators in our CFF Uniswap model. Line 1 in~\fig~\ref{fig:uniswap_contract} then corresponds to ``A", or the initial configuration of our model when this semantic rule applies. This semantic rule describes execution when the next instruction to execute (first transaction in the \kcell cell, wrapped in an ``exec" keyword) is a token swap on Uniswap for swapping a symbolic amount \textsc{TradeAmount} of the input token denoted by symbol \textsc{TokenIn} for an output token denoted by symbol \textsc{TokenOut}. This swap rewrites to (``$=>$") a series of statements (Lines 2-6) that will execute one at a time with separate operational semantic rules in CFF. The ellipsis in the \kcell cell signifies the remaining transactions, in \scell cell signifies the rest of the state mapping, and in the \bcell cell signifies the prefix of the block constructed so far.

We leave further exploration of our executable models to the interested reader, and provide more notes on K-specific keywords in the above model in \iffull Appendix~\ref{sec:morek}\else the expanded version of this work~\cite{fullversion} \fi.  We also describe some refinements necessary for a model that behaves the same as deployed DeFi contracts and discuss subtleties of modeling MakerDAO liquidations \iffull in Appendix~\ref{sec:refinements} \else in~\cite{fullversion}\fi.
\else
We provide a more complete explanation of the syntax and semantics of this model in Appendix~\ref{sec:moreuni}.
\fi


\section{Experimental Evaluation}
\label{sec:experiments}


Using our full CFF models (not the simplified ones from above), we ran several experiments on data from Uniswap V1, Uniswap V2, SushiSwap, and MakerDAO, which we detail here. We aim to experimentally address several key questions:
\begin{enumerate}[noitemsep,leftmargin=*]
\item Are our CFF models accurate in reproducing the on-chain behavior of corresponding contracts? How efficient is this execution?
\item Can our models yield mechanized proofs about the extent of security of DeFi contracts and their composition while handling transaction reorderings and generic transaction insertions by miners?
\item Is use of our CFF models economically sensible in uncovering DeFi exploits on-chain?

\end{enumerate}

\mypara{Experimental setup.}
We ran most of our experiments on a mid-range server, equipped with an AMD EPYC 7401P 24-core server processor, 128GB of system memory, and a solid-state drive. For our computations, only the result is written to disk, and therefore our code is primarily CPU-intensive. We did not observe substantial memory overhead. For our parallelism experiments only, we used an AWS cluster of c5 instances with 256 vCPUs unless specified otherwise.

\mypara{Dataset collection.}
We used Google's BigQuery Ethereum to download every swap and liquidity event generated (until May 16, 2021) by Uniswap V1, Uniswap V2, and SushiSwap. These are three Uniswap-like AMMs that see substantial volume and are relevant to our analyses. In total, we collected 50,038,981 swaps, 2,317,917 liquidity addition events, and 844,709 liquidity removal events traded for 39,329 token pairs. For each token pair, we created a chronological log of events.

For MakerDAO, we used BigQuery to download all the log events generated (until May 16, 2021) by its core smart contract\footnote{https://github.com/makerdao/dss/blob/master/src/vat.sol} which manipulates CDPs (``vaults'') and updates stability fees and oracle prices. This data includes 322,771 CDP manipulation events (including 284 liquidations) across 18,642 CDPs and 25 collateral types. For each collateral type, we created a chronological log of all relevant events.

\subsection{Execution Validation and Performance Experiments}
\label{subsec:expvalidation}
We start with experiments to validate our CFF models with on-chain data and show the performance of our CFF tool.

\mypara{CFF model validation.}
We executed our CFF models on the collected data to ensure that our framework computes the correct final state, i.e., actual on-chain state. For the data from the three AMMs, we ran our executable semantics and inspected the resulting chain. We found that the resulting chain state from our CFF models matches exactly the on-chain state.  

We evaluated our CFF Maker model similarly. We found that the stability fees and final debt and collateral values for each CDP before liquidation exactly match the chain state. Since we do not model the liquidation auction mechanism, we do not expect the Maker model to accurately derive the state after liquidation events. MEV reported in our experiments only depends on the state before the first liquidation. The state after liquidation does not affect our results.

We provide scripts to download, process, and validate data for each protocol in the \texttt{all-data}~sub-folder of our repository. This validation mechanism highlights the importance of executable formal semantics: execution is a key requirement for validating abstract formal models against real-world data.

\mypara{CFF performance and parallelism.}
We evaluate the performance for two types of functionalities. First, for different UniswapV2 token pairs, we execute all corresponding on-chain transactions that manipulate the state in the same order as they happened. This measures the execution time of our model, or the time to derive the full on-chain state from the list of transactions. \fig~\ref{fig:time-txs} shows the time taken for our CFF to derive the state for different pairs as a function of the number of transactions executed for the pair. K's internal execution engine intrinsically gives roughly a 4x parallel speedup, which can be seen in the figure as a speedup of real/wall execution time over the amount of total CPU time required to compute model state.
These results, combined with our model validation, answer our first experimental question. Our modeling execution engine is sufficiently performant to ensure that our models' output matches the full chain state on Ethereum for all relevant transactions using only commodity hardware. For instance, the most active pairs traded on any AMM contained about 100k transactions in our data, and it took under 2 hours of CPU time to parse this data and perform end-to-end model validation.

Second, we evaluate the performance for exploring all possible reorderings available to a miner as part of their extraction of MEV, and analyze how the computation of optimal miner orderings can be efficiently parallelized.  This will allow us to use our models to also find transaction orderings not exploited by past miners. For these experiments, we use an AWS \texttt{c5.metal} instance optimized for computation.  This machine features 96 3.9 GHz cores running on Intel’s Second Generation Xeon Cascade Lake processors, with 192GiB of available memory. 
In \fig~\ref{fig:time-threads}, we report the average execution times for attacks with 7, 8 and 9 transactions to be reordered using different number of CPU cores. As discussed in Section~\ref{subsec:amm_experiments}, blocks with 10 or more relevant transactions (i.e., transactions interacting with our models) are rare. Transactions chosen for this particular figure are UniswapV2 transactions and MakerDAO transactions explored using a composition of our UniswapV2 and MakerDAO models, so as to be representative of our MEV extraction experiments described in Section~\ref{sec:composition-experiments} ; changing to a different transaction type that deals with our other models does not have any material impact on the reported numbers.
%
Since we used a 96-core machine for our experiments, and given that K provides a 4x parallel speedup, we find that the real wall clock time converges to the fastest execution speed at around 24 worker threads before CPU limitations are reached.  Given that our parallel exploration of possible state spaces has no synchronization between parallel workers, the embarrassingly parallel nature of this problem suggests future scaling across machines to be a natural direction for handling larger problem instances.  Before the scale ceiling of 24 parallel workers is hit, approximately linear scaling is visible in \fig~\ref{fig:time-threads}, with some overhead associated with scheduling threads and managing shared system resources.



\begin{figure}[t!]
\centering
\includegraphics[scale=0.45]{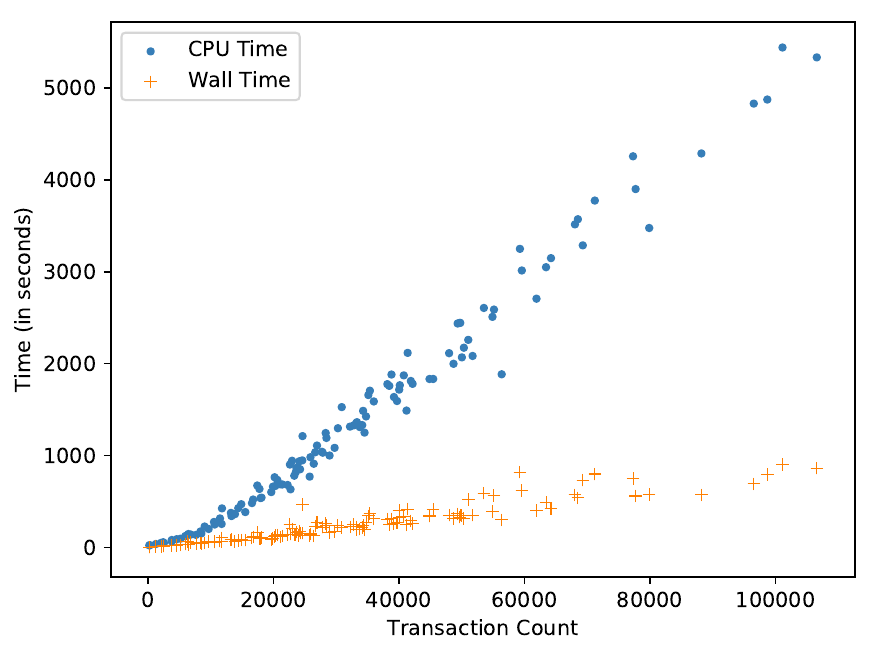}
\vspace{-0.5em}
\caption{CFF execution time to evaluate and validate resultant state for a transaction sequence.}
\label{fig:time-txs}
\iffull \else \vspace{-2.3em} \fi
\end{figure}

\begin{figure}[t!]
\centering
\includegraphics[scale=0.45]{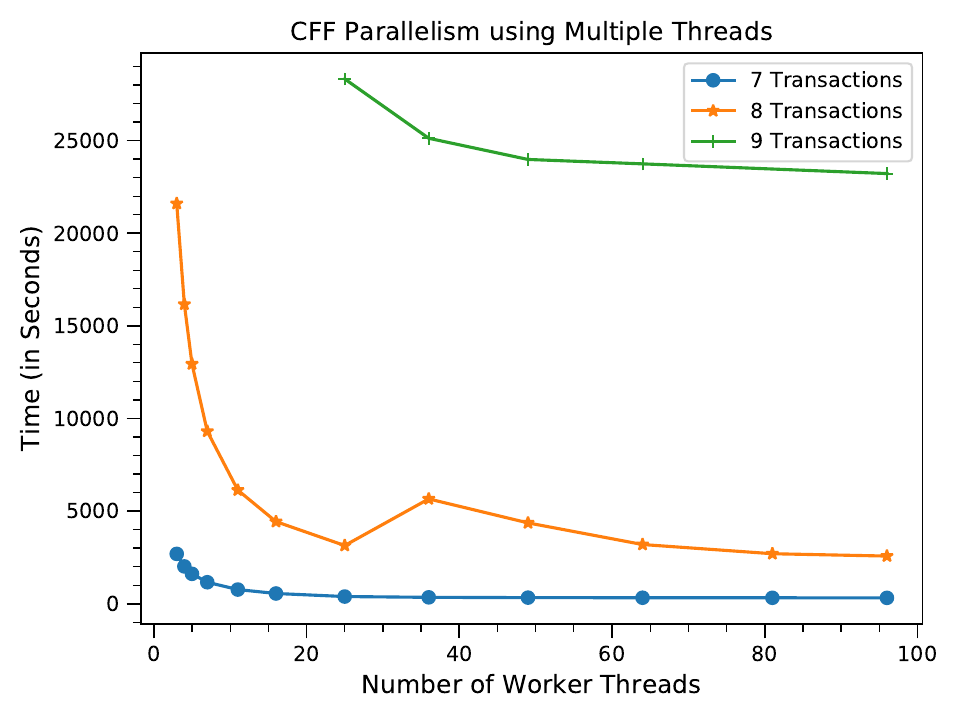}
\caption{CFF Parallelism: Time taken to explore all reorderings with varying number of transactions (7,8,9) as a function of the number of threads used.}
\label{fig:time-threads}
\iffull \else \vspace{-2em} \fi
\end{figure}

\subsection{Mechanized Proofs and Symbolic Invariants}
\label{sec:proof-example}
We now use the deductive program verifier (\texttt{kprove}) from the K framework along with our refined CFF models to assess the security of the composition of Sushiswap and UniswapV2. To achieve this, we have to specify the initial state of the two contracts along with the set of transactions interacting with these particular contracts. These transactions include the user transactions as well any given symbolic transactions inserted by the miner. We also specify a reachability claim that MEV is no greater than 0. If the two contracts compose securely as per our definition in Section~\ref{sec:model}, then running \texttt{kprove} generates a deductive proof for the specified claim. On the other hand, when the composition under the specified initial state is insecure, \texttt{kprove} automatically generates a counterexample strategy (i.e. sequence of transactions) and a symbolic invariant for MEV in terms of the symbols appearing in the initial state or inserted transaction template. More precisely, the symbolic invariant is a set of (satisfiable) formulae representing the amount of MEV in terms of the variables appearing in the specified initial state and the transactions applied to it.

While our CFF can reason about the security of any specification of initial state and set of transactions, we describe an example detailed specification in Appendix~\ref{sec:kproof} capturing one of the biggest arbitrage opportunities\footnote{\href{https://etherscan.io/tx/0x2c79cdd1a16767e90d55a1598c833f77c609e972ea0fa7622b70a67646a681a5}{0x2c79cdd1a16767e90d55a1598c833f77c609e972ea0fa7622b70a67646a681a5}} observed on-chain involving two AMMs as reported in~\cite{qin2021quantifying}. To capture this arbitrage opportunity, we specify the AMM states at blocknumber 10854887, the user transactions interacting with the AMMs, and swap transactions inserted by miner with symbolic parameters (representing the size of miner's trade). We plot the MEV formula output by our CFF representing the available MEV opportunity as a function of the size of the trades inserted by the miner in Appendix~\ref{sec:kproof}. The arbitrageur in this arbitrage made a profit of 76 ETH, while our CFF reports a higher MEV of 123 ETH {\em not captured by miners}.

This example illustrates the power of CFF in finding opportunities left on the table by arbitrageurs currently.  Note that our refined mechanized models account for fees, slippage, and integer rounding and hence, the size of the opportunity available to the miner is slightly less than the theoretical value derived in Section~\ref{sec:composability}.
We provide the full specification in \texttt{proofs} sub-folder of our repository. CFF can also mechanically reason about the security of many AMMs composed together, as well as more complex composed smart contracts, but we leave this to future work.

\subsection{AMM Experiments}
\label{subsec:amm_experiments}
We ran a series of experiments on our CFF models for the three AMMs to quantify the MEV extractable from them, and prove the utility of our models further by furnishing real-world insights into available MEV. Our experiments are intended to validate the ability of our tool to uncover profit-seeking miner strategies, and can easily be used for other DeFi contracts.


\iffull\begin{figure}[t!]\else\begin{figure}[t!]\fi
    \centering
    \includegraphics[width=0.85\columnwidth]{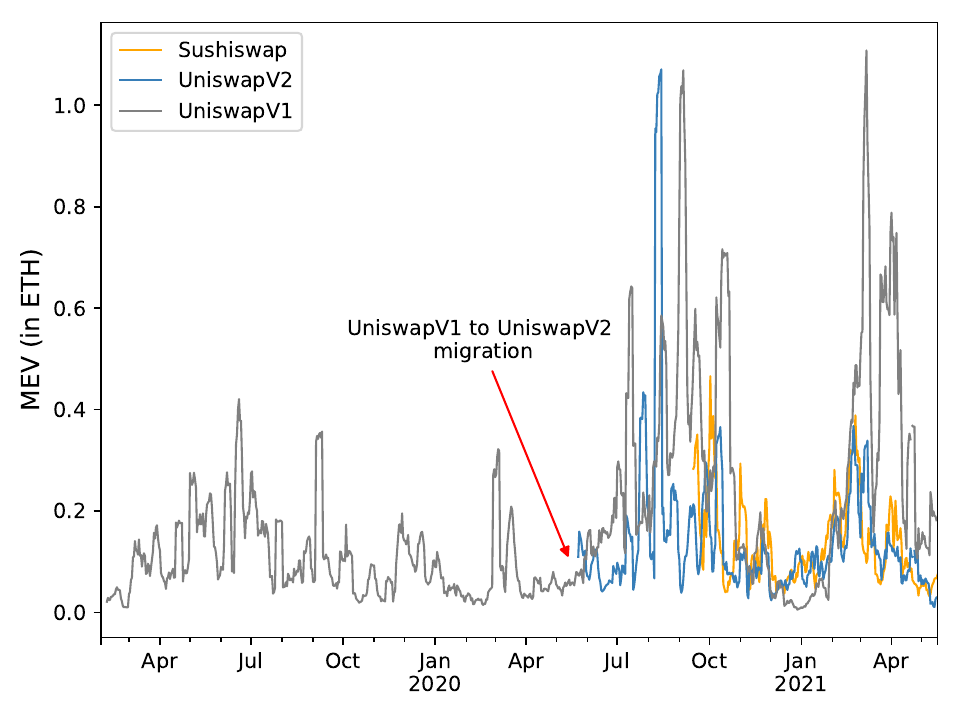}
    \iffull \else \vspace{-0.25em} \fi
    \caption{7-day moving average of MEV per block in a random sample of 1000 random blocks in each month. 1 ETH $\sim$ 3200 USD at the time of writing.}
    \label{fig:random-mev}
    \iffull \else \vspace{-2em} \fi
    \iffull\else\end{figure}\fi
    \iffull \end{figure} \fi

\iffull\begin{figure}[t!]\else\begin{figure*}[ht]\fi
\centering
\begin{subfigure}{\triplefigwidth}
\includegraphics[width=\columnwidth]{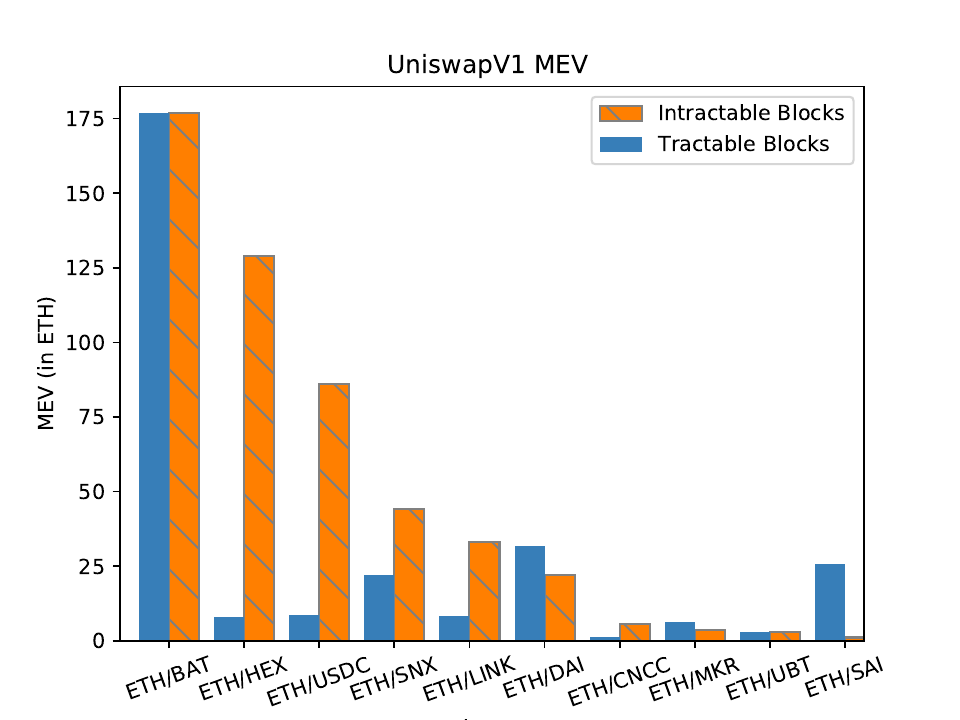}%
\caption{Uniswap V1 MEV}
\label{fig:uniswapv1_mev}
\end{subfigure}\hfill
\begin{subfigure}{\triplefigwidth}
\includegraphics[width=\columnwidth]{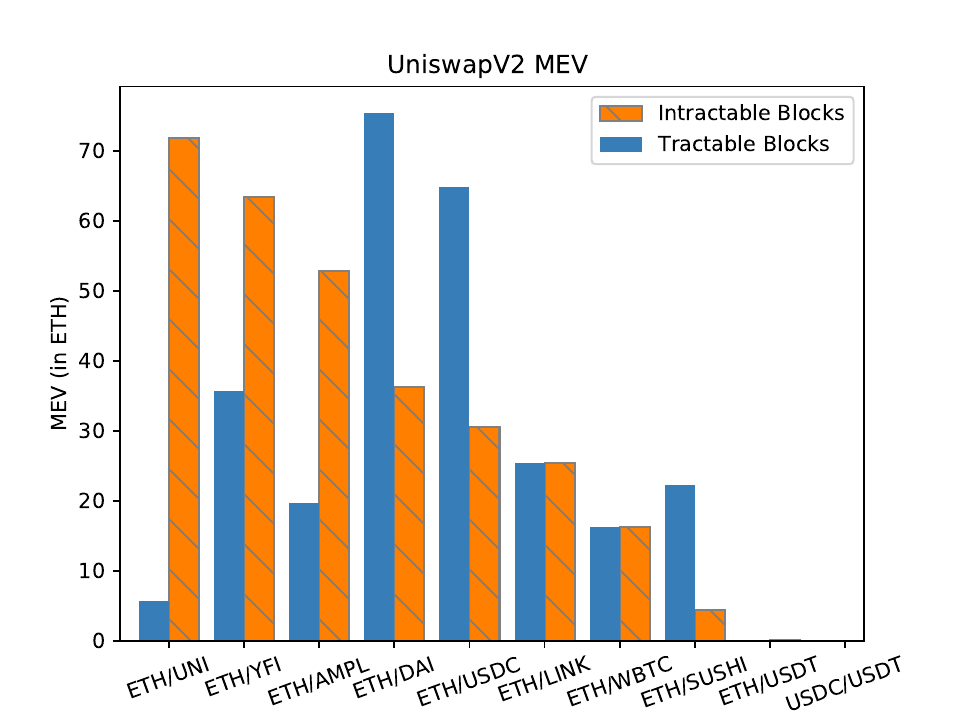}%
\caption{Uniswap V2 MEV}
\label{fig:uniswapv2_mev}
\end{subfigure}\hfill
\begin{subfigure}{\triplefigwidth}
\includegraphics[width=\columnwidth]{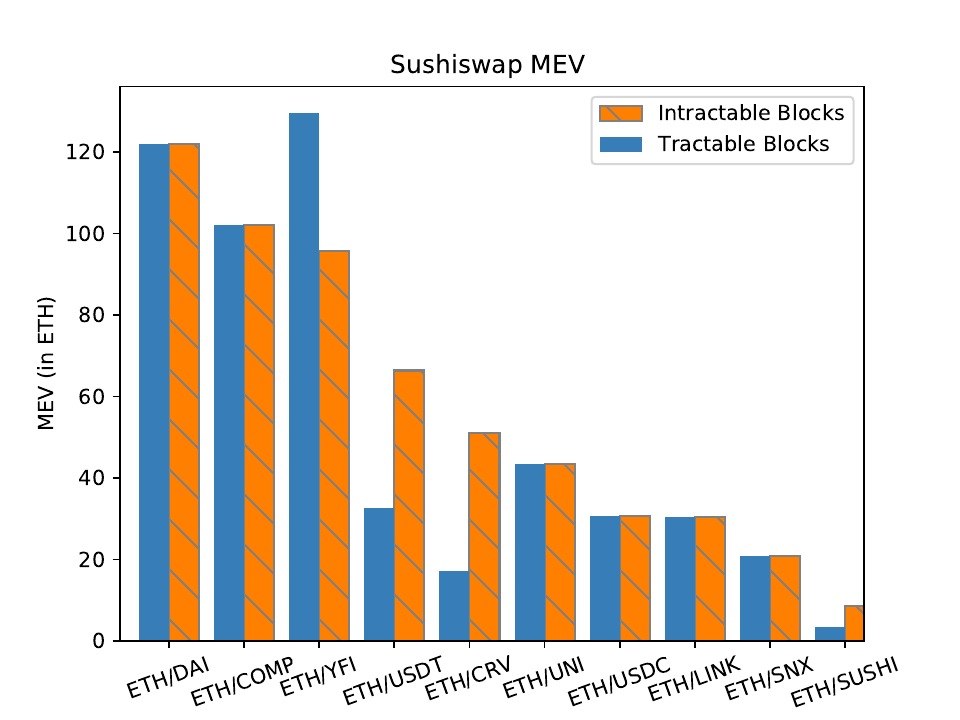}%
\caption{Sushiswap MEV}
\label{fig:sushiswap_mev}
\end{subfigure}
\caption{Highest observed MEV blocks for the top 10 most active token pairs in our dataset. Intractable blocks have 10 or more transactions involving the pair, and are partially explored by our tool through a random search.}
\label{fig:top10-mev}
\iffull\else\end{figure*}\fi
\iffull \end{figure} \fi

\iffull\begin{figure}[t!]\else\begin{figure*}[ht]\fi
    \centering
    \begin{subfigure}{\triplefigwidth}
    \centering
        \includegraphics[width=\columnwidth]{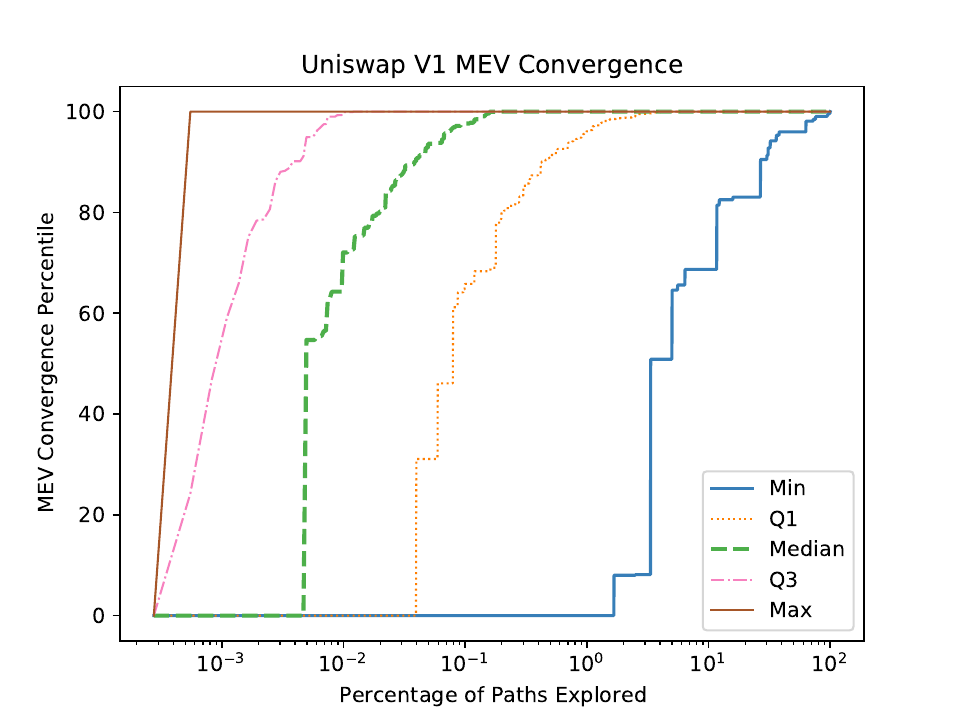}
        \caption{Uniswap V1 Convergence}
        \label{fig:uniswapv1_convergence}
    \end{subfigure}
    \hfill
    \begin{subfigure}{\triplefigwidth}
    \centering
        \includegraphics[width=\columnwidth]{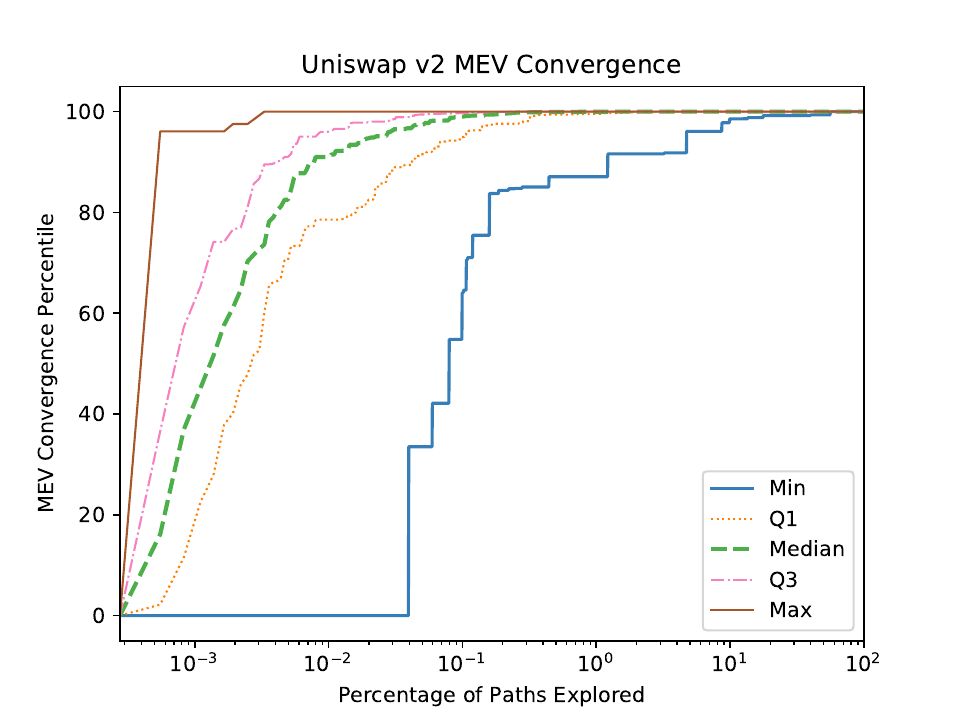}
        \caption{Uniswap V2 Convergence}
        \label{fig:uniswapv2_convergence}
    \end{subfigure}
    \hfill
    \begin{subfigure}{\triplefigwidth}
    \centering
        \includegraphics[width=\columnwidth]{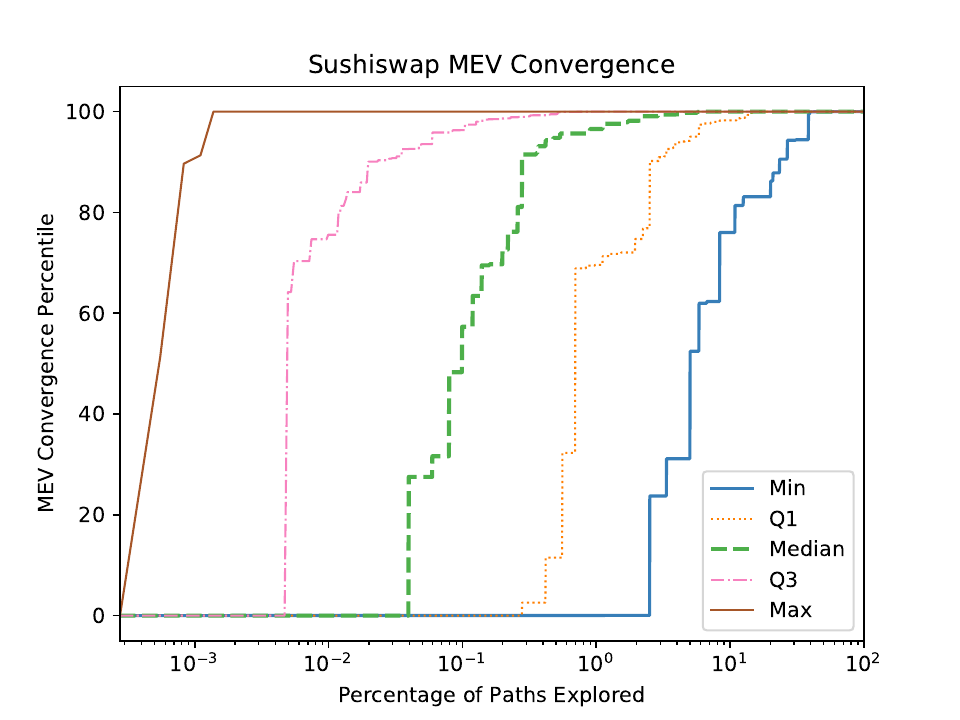}
        \caption{Sushiswap Convergence}
        \label{fig:sushiswap_convergence}
    \end{subfigure}
    \caption{Convergence towards the optimal MEV for a random sample vs percentage of total paths explored for tractable blocks.}
    \label{fig:approx-convergence}
    \iffull \else \vspace{-2em} \fi
\iffull\else\end{figure*}\fi
\iffull \end{figure} \fi

\mypara{Reordering to lower-bound MEV.}
We consider all possible transaction reorderings that can be performed by a miner. For this, we do not consider transaction insertion by miners, and therefore we will find a lower bound on the MEV by computing the difference between the most and least profitable transaction ordering with respect to a user who colludes with the miner to get the most profitable ordering. Otherwise stated, we define MEV in this setting as the amount a miner could make with a composed ordering bribery contract. We expand on this subtle difference in Appendix~\ref{sec:ammbounds}. In certain cases of restrictions imposed by other (wrapper) contracts involved in the transaction, not all reorderings might be valid. We automatically validate the optimal ordering in the last phase of CFF as described in Section~\ref{subsec:design_impl}. Note that providing our CFF tool with the models of the other (wrapper) contracts interacting with the AMM contracts would ensure that this validation is unncessary, however we defer this to future work.

For each AMM that we support, we conduct two kinds of analysis: First, we analyse the average MEV in a randomly sampled block (having transactions for any token pair) obtained by sampling 1000 blocks per month that have at least 2 transactions interacting with it. We report the 7-day moving average of MEV found per block as a time series plot in \fig~\ref{fig:random-mev}. For the year 2021, total MEV across all the AMMs in our random sample is 1.5 million USD, which by extrapolation comes to about 56 million USD per month in 2021. Second, we examine the token pairs with the top 10 highest number of transactions, and randomly sample 30 blocks involving these token pairs. Our tool can fully explore the state space for blocks with 9 or fewer AMM transactions; we call these blocks ``tractable''. We report the average MEV found per block (for each token pair) in our random sample in \fig~\ref{fig:top10-mev}.

\mypara{Intractable-block exploration.}
For blocks with 10 or more relevant AMM transactions (i.e., transactions that interact with the AMM), we do not explore the full search space. Instead, for these ``intractable blocks,'' we compute the MEV through a randomized search. We explore 400,000 paths, but randomize which paths are explored. The average MEV values for intractable blocks in our random sample are also reported in \fig~\ref{fig:top10-mev}. Because our primary aim was developing and validating our models’ ability to find attacks, we did not optimize this search for performance further. Using further optimization or more parallel computation could likely yield more accurate estimates for intractable blocks, but we defer this to future work. We found that ``intractable'' blocks are rare in our dataset. \fig~\ref{fig:txcount} shows a histogram of the number of blocks containing a particular number of AMM transactions.

\mypara{Approximate convergence.}
To support our exploration of intractable blocks, a natural question is to what extent a random search on a sample of orderings approximates the MEV for a given block. For this, we look at how the MEV converges for tractable blocks as more paths are explored iteratively. For each AMM, we randomly explore the same tractable blocks in our random sample, and report the quartiles for MEV convergence in \fig~\ref{fig:approx-convergence}. On average, we uncover 70\% of MEV in more than 90\% of the instances by exploring just 1\% of total paths. Since we explore 400,000 paths for intractable blocks, we explore roughly 11\% of the total paths for blocks with 10 transactions, and roughly 1\% of the total paths for blocks with 11 transactions. As evident from \fig~\ref{fig:txcount}, blocks with more than 11 transactions are even more rare. 

\begin{figure}[t!]
\centering
\includegraphics[width=0.75\columnwidth]{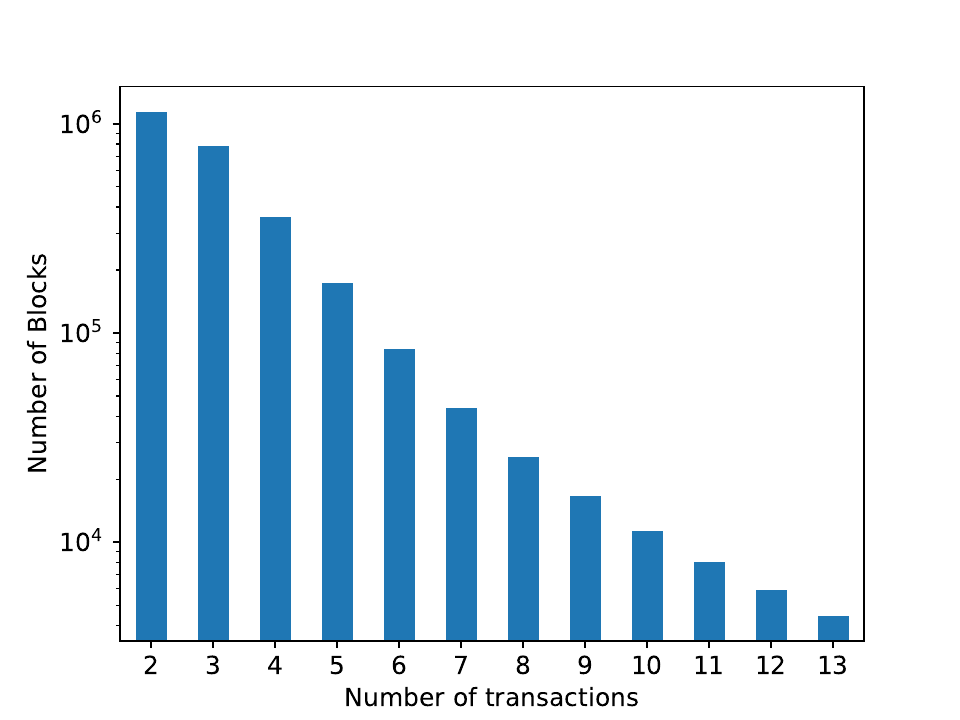}
\caption{Distribution of AMM transactions in blocks}
\label{fig:txcount}
\iffull \else \vspace{-2em} \fi
\end{figure}

\mypara{Reordering insights.}
Our results show that UniswapV2 exposes significantly less MEV compared to UniswapV1 and Sushiswap, thanks to the huge liquidity on UniswapV2. It is interesting to note that some of the token pairs have negligible MEV compared to the rest. It turns out that all of these pairs include a stablecoin (or are both stablecoins, e.g., USDC/USDT), which exposes only small price fluctuations for users across reorderings. On the other hand, pairs with unstable prices (UNI, YFI, BAT) expose the highest MEV (75-175 ETH). On manual examination, we find that the blocks exposing huge MEV ($\sim$100 ETH) often involve a user making a big purchase of token X with token Y and being either frontrun or backrun by a bot. In \iffull Appendix~\ref{sec:deepdive} \else the full version of this work~\cite{fullversion}\fi, we provide a deep dive into a backrunning example---one of the highest MEV instances uncovered by our tool.





\subsection{Composability Experiments}
\label{sec:composition-experiments}
\begin{figure}[t!]
\centering
\includegraphics[width=0.75\columnwidth]{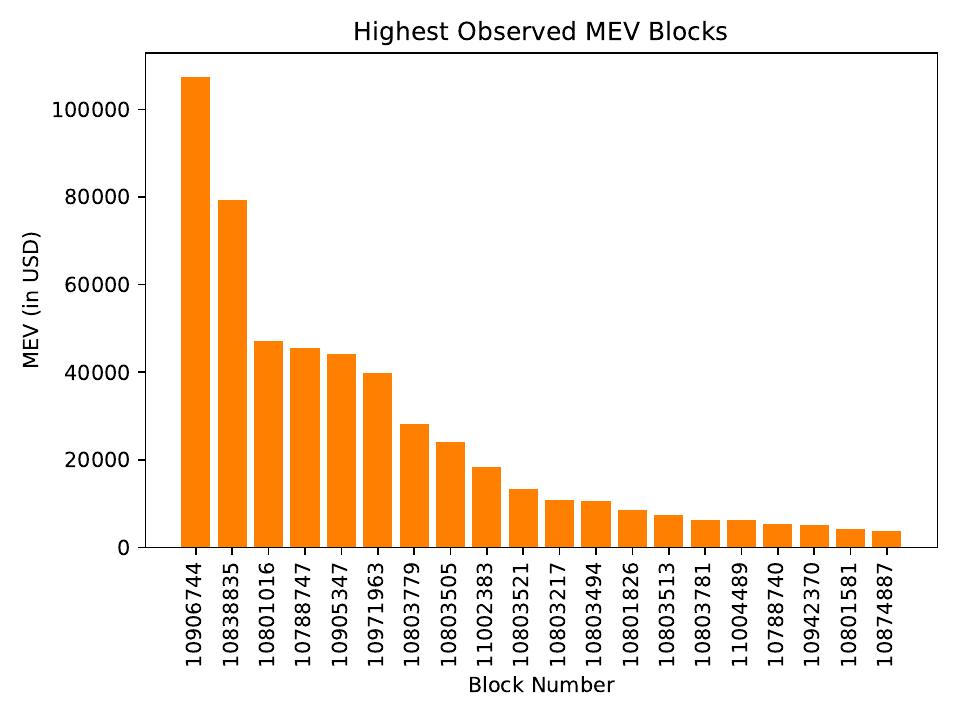}
\caption{MEV for Maker composed with Uniswap V2}
\vspace{-1.5em}
\label{fig:maker-mev}
\end{figure}

To highlight the capability of our tool in finding MEV in the composition over multiple contracts, we consider our running example of the composition between MakerDAO and Uniswap. Here, we use the price from Uniswap V2 instead of the one from Maker's oracle module. Although MakerDAO does not currently use Uniswap as a price oracle, making the attacks in this section purely theoretical, this change reflects similar proposals from over 60 projects (enumerated at \url{https://debank.com/ranking/oracle}), as well as academic results suggesting a possible security argument for such a change~\cite{angeris2020improved}. Using our tool, we can compute the MEV exposed as a result of MakerDAO adopting this potential composition.

\mypara{Oracle attacks.}
We extend the AMM reordering experiments from Section~\ref{subsec:amm_experiments} to allow for an additional miner action, where the miner can liquidate under-collateralized CDPs. Formally, if CDPs with index $1,...,n$ are open in the system, the set of transactions $s$ is extended to include a liquidation of all $n$ CDPs by a miner account $M$. We then compute the total amount of profit earned by $M$ from any successful liquidations as a lower-bound metric for MEV.

To quantify this, we examine on-chain data for the top 100 CDPs and blocks in MakerDAO when the CDPs are at the highest risk of liquidation (i.e., CDPs with the least collateral-to-debt ratio). For a given block, we consider possible reorderings over all Uniswap V2 and MakerDAO transactions, and then compute the MEV as a result of a miner inserting a CDP liquidation transaction. We report this in \fig~\ref{fig:maker-mev} for the top 20 blocks with the largest liquidations (calculated using the collateral value at the time of liquidation). We found a total MEV of 542,827 USD---orders of magnitude larger than the block rewards and transaction fees for these blocks.  These experiments can be reproduced using the \texttt{run\_mcd\_experiments} script in our Github repository.

\subsection{Other Notable Attacks}
\label{sec:other-attacks}

\mypara{Airdrops.} Airdrops are a recent DeFi phenomenon where users who have taken a specific action on the blockchain (e.g., interacted with some contract function, held an NFT etc.) can claim a proportionate share of a newly released token. If the airdrop contract checks only the ownership in the current state and not the historical record, then it can be exploited using flash loans. One such exploit was observed recently where an attacker was able to exploit the much anticipated ApeCoin airdrop for BAYC NFT holders for approximately \$1,100,000~\cite{bayc-exploit}
\footnote{\href{https://etherscan.io/tx/0xeb8c3bebed11e2e4fcd30cbfc2fb3c55c4ca166003c7f7d319e78eaab9747098}{0xeb8c3bebed11e2e4fcd30cbfc2fb3c55c4ca166003c7f7d319e78eaab9747098}}. We reproduce this attack using CFF. To this end, we implement 3 new CFF models. First, a flash loans model that has a rewrite rule (with appropriate state updates) for allowing any player to borrow desired amount of a certain fungible token, call another contract, and then deposit back certain amount of the same fungible token. The rule \emph{requires} that the deposited amount be greater than the borrowed amount along with some fees. The second model is for a ``vault'' contract that allows for minting and redeeming of fungible tokens (``BAYC tokens'' here, which function as a fungible wrapper to the BAYC NFTs) against NFTs pooled together in a vault. The third model is for the na\"ive aidrop contract that allows any player to claim a fixed amount of ApeCoin tokens against their NFT for which a claim has not been passed before. We compose these models along with our Sushiswap model in CFF in order to obtain a strategy (counterexample to composability proof) that yields the same amount of profits in ETH as observed in the attack~\cite{bayc-exploit}. The strategy first borrows BAYC tokens through the flash loans model, calls into the vault model to redeem them for other players' NFTs found in the vault, claims the ApeCoin airdrop for these NFTs, then returns the NFTs back to the vault for the BAYC tokens which it pays back to the flash loans model with fees. Finally, the ApeCoin tokens are swapped on Sushiswap for ETH.

\mypara{Governance.}
We use CFF to illustrate how flash loans can be used to exploit governance mechanisms. To this end, we model a simple governance contract that finalizes the vote at a certain blocknumber based on the capital staked for or against the vote in the current state. As a proxy for the economic incentives from the governance vote, we model a simple betting contract (conceptually similar to $\betcontract$) that awards any player a certain amount of ETH if the vote passes. We use CFF to study the composability of the flash loans contract, the governance contract, and the betting contract. The current state supplied to CFF has symbolic variables $x$ for the flash loans reserves, $y$ for the capital staked in favor of the vote and $z$ for the capital staked against the vote. CFF outputs a strategy (counterexample to the composability proof) with the MEV equal to the betting contract reward less the flash loans fees, along with the condition:
\begin{lstlisting}[numbers=none]
(x > z - y) and (x > 0 ) and (y >= 0) and (z >= 0)
\end{lstlisting}
\vspace{-1.75em}
We provide the models for the flash loans, vault, airdrop, governance and betting contracts used above in the \texttt{cff\_models} directory, and provide the modules for reproducing the Airdrops attack and Governance attack in the \texttt{proofs} directory of our Github repository.



\section{Conclusion}

We have introduced a powerful and novel approach---that adopts the lens of miner-extractable value (MEV)---for reasoning about and quantifying security guarantees for DeFi contracts and their interaction. We have instantiated a number of semantic models in a new computational framework, the Clockwork Finance Framework (CFF)---an executable proof system that allows us to reason about the financial security of smart contracts.  We have provided open-source models, both abstract and executable, that represent key MEV-exposing deployed smart contracts.  We have shown how our definitions enable powerful proofs of composition for popular smart contract protocols, a missing ingredient in the current  deployment of DeFi contracts.  We believe that MEV, smart contract composition, and formal verification can serve as viable key ingredients for empirically and rigorously measuring and improving DeFi contract security.

\iffull \section*{Acknowledgments}
\else \mypara{Acknowledgments.} \fi
We thank Alexander Frolov for contributing to the AWS infrastructure needed to scale our experiments.
This work was funded by NSF grants CNS-1564102, CNS-1704615, and CNS-1933655 as well as generous support from IC3 industry partners. Philip Daian is a co-founder of Flashbots, a research and product organization developing solutions related to MEV, and has financial interests in several decentralized exchange protocols. Any opinions, findings, conclusions, or
recommendations expressed here are those of the authors and
may not reflect those of these sponsors.


\ifCCS
\bibliographystyle{ACM-Reference-Format}
\bibliography{references}
\fi

\ifUSENIX
\bibliographystyle{plain}
\bibliography{references}
\fi

\iffull \printbibliography \fi

\ifIEEE\else \appendix \fi
\iffull \else

\fi

\section{DeFi Exploits Background}
\label{appendix:defisploit}

\paragraph{Attacks and arbitrage.}
A number of practical issues arise in the deployment of secure DeFi systems. The first carries over from traditional software systems, since the guarantees upheld by any financial instruments are only as good as the software that underlies them. Both smart contracts in general and DeFi instruments have seen a wide number of software failures that eroded their security guarantees (e.g.,~\cite{atzei2017survey, qin2020attacking, qureshi2017hacker}), as well as corresponding academic and practical interest in rectifying these failures (e.g.,~\cite{breidenbach2018enter, chen2020survey, grishchenko2018semantic, hildenbrandt2018kevm, krupp2018teether}).

Further security concerns occur when the intended design and functional guarantees provided by DeFi systems do not align with the financial guarantees desired by users. For instance, in decentralized exchanges, while some users may assume that ``secure" exchange implies fair execution of their submitted orders, recent work~\cite{daian2020flashboys} finds that the core system design itself could prove unfair to its users, detailing how inefficiency can be exploited by arbitrageurs to introduce systematic security failures. In many instances, because ``attacks'' involve exploiting inefficiencies in DeFi systems in unexpected ways to profit financially, it makes sense to express security properties from an economic standpoint. Consequently, for DeFi systems, the distinction between profit-seeking techniques like arbitrage, and security failures can become difficult to draw.

Sometimes, the distinction is obvious. A typographical error in the program of a smart contract (as described in~\cite{buterin2016thinking}) is one common source of funds-loss that can be clearly categorized as a ``code bug."  Similarly, a decentralized exchange that is designed to allow arbitrage by programmatic actors analogous to those in traditional financial exchanges can be clearly classified as ``financial arbitrage"~\cite{daian2020flashboys}, and therefore not a security vulnerability. 

Unfortunately, many DeFi exploits fall less clearly into either category. A noteworthy example is the string of recent high-profile attacks on the bZx DeFi protocol, which relied on data from several external DeFi instruments. An attacker was able to break the invariants of the external contracts and use them to profit from the bZx protocol. Now, invariant checks could easily have been done within the bZx protocol code, in which case the root cause would be a software failure of the bZx contract. At the same time, the exploit could also be viewed as a design flaw since it is impossible to determine during the execution of a DeFi transaction, whether the external feed has been manipulated through arbitrage. In Appendix~\ref{appendix:bZx}, we use the bZx attacks as a case study to understand how the distinction between software exploits and fundamental design failures manifests in the real world.

\subsection{Case Study on $\mathrm{bZx}$ attacks}
\label{appendix:bZx}
In this section, we use the attacks on the bZx protocol to understand the nuances between security vulnerabilities in smart contracts, and arbitrage-like design choices. The bZx protocol was originally designed to allow decentralized margin trading and lending, and was the target of recent high-profile attacks. The core of one of these attacks was the ability of a malicious attacker to use flash loans to perform a massive short in the bZx protocol. The bZx contract relied on Uniswap, a decentralized exchange, to sell coins at what it assumed was market price. But, because the size of the attacker's flash-loan-based short order exceeded the amount that could be safely traded using the liquidity in the Uniswap exchange, the short increased the price of wBTC (wrapped Bitcoin) tokens on the Uniswap platform for this transaction. The attacker was then able to use this false rate to borrow wBTC against ETH (Ethereum, the native currency of the Ethereum blockchain), selling the newly borrowed wBTC into this falsely inflated price and obtaining ETH profit. A comprehensive post-mortem expos{\'e} summarizing the attack is available in~\cite{bzxpostmortem}.

%
The bZx attack blurs the line between arbitrage and code vulnerabilities. One can easily view the failure of the bZx protocol to check that the Kyber/Uniswap order routing had sufficient liquidity to complete its order as a code failure in the bZx protocol, in which case the attack represents a software exploit against bZx. But, one can also view this failure as a fatal design flaw, as it is impossible to determine during the execution of a DeFi transaction whether a given price represents the true market price outside the system in which the transaction is executing, in which case the attack represents financial arbitrage that more closely resembles the activity in traditional financial markets when inefficient financial products operate as intended.

This debate is not purely theoretical. One DeFi insurance product, Nexus Mutual, insured users of the bZx protocol against losses stemming from failures in the correct operation of the underlying smart contracts.  Nexus Mutual however did not cover issues in design, and would not need to pay out to its users if the smart contracts were deemed to be operating as intended.  After some debate, the Nexus Mutual fund decided to pay out to users who lost money in the bZx attack, as they reasoned that the bZx smart contract designers \emph{intended} to check the slippage the attacker took advantage of, and the attacker bypassed this check due to a coding error~\cite{bzxnexus}. While in this case the Nexus Mutual operators were able to come to a determination, we expect that future DeFi attacks will continue blurring the line between design and implementation issues, especially at the interfaces between various composable interoperable financial components.  In a DeFi context, both types of attacks can be viewed as a programmatic search for a reachable final state in the system in which the attacker profits.  The attack is far from unique; for example, just a few months later, \$100M was drained from a similar protocol in a similar exploit pattern~\cite{compoundsploit}.

\section{Generalized MEV and Composability Definitions}
\label{appendix:GMEV}
In Section~\ref{sec:model}, we defined $\kMEV$ which computes the $\MEV$ for a miner if it appends $k$ consecutive blocks to the chain and can change the transaction ordering across those $k$ blocks. In this section, we define weighted miner-extractable value, or $\WMEV$, which is weighted by the probability that a miner can mine $k$ consecutive blocks.

Formally, for a miner $P$, let $p_k$ be the probability that it mines exactly $k$ consecutive blocks. We assume that $p_k$ is not state dependent (at least in the short term). $p_k$ may be a function of the mining difficulty or the fraction of hash power owned by the miner. We can now define weighted $\MEV$ as:

\begin{definition}[Weighted MEV]
\[
    \WMEV(P,s) = \sum_{k = 1}^{\infty} p_k \cdot \kMEV(P,s)
\]
\end{definition}

\noindent As a simple example, consider a miner $P$ who controls a fraction $f$ of the total hash power. If we assume that mining is modeled as a random oracle and that there is no selfish mining, then the probability $p_k$ that $P$ mines exactly $k$ consecutive blocks is $p_k = f^k(1-f)$. Suppose further that the extra $\MEV$ obtained per extra mined block is a constant $m$. For this simplified example, we can compute the $\WMEV$ as:
\[
    \WMEV(P,s) = \sum_{k=1}^{\infty} f^k (1-f) (km) = \frac{fm}{(1-f)}
\]

Equipped with this, we can also generalize  the definition of Defi composability to include $\WMEV$. For this, $\MEV$ in Definition~\ref{def:single_defi_compose} will be replaced by $\WMEV$.

\paragraph{Miner cost.}
All of our notions of extractable value abstract out the actual cost incurred by the miner (e.g., the cost of equipment, electricity). We do this to make our definitions more broadly applicable. We note that the cost of a specific miner can be calculated independently, and subtracted from the extractable value to obtain the profit a miner could make from transaction reordering.

\iffull \else
\section{Additional Composability Examples}
\label{sec:morecomposability}

We further explore composability examples here. 

\subsection{Maker Contract Model}
\label{subsec:maker}

The Maker protocol allows users to generate and redeem the collateral-backed ``stablecoin" Dai through Collateralized Debt Positions (CDPs). Users can take out a loan in Dai by depositing the required amount of an approved cryptocurrency (e.g., \ETH) as collateral, and can pay back the loan in Dai to free up their collateral. If a user's collateral value relative to their debt falls below a certain threshold called the ``Liquidation Ratio" ($>1$), then their collateral is auctioned off to other users in order to close the debt position. Maker uses a set of external feeds as price oracles to determine the value of the collateral. A separate governance mechanism is used to determine parameters like the Liquidation Ratio, stability fees (interest charged for the loan), etc., and also to approve external price oracle feeds and valid collateral types. We consider here a simplified version of Maker's single-collateral CDP contract that does not model stability fees, or liquidation penalties. The contract $\makercon^{(\mathbf{X}, \mathbf{Y})}$ allows users to take out (or pay back) loans denominated in token $\mathbf{X}$ by depositing (or withdrawing) the appropriate collateral in token $\mathbf{Y}$, and allows for liquidation as soon as the debt-to-collateral ratio drops below the Liquidation Ratio. \fig~\ref{fig:maker_contract} details the contract.

It should be noted that the amount of collateral liquidated and received by the liquidator as well as the debt (in Dai) paid off by the liquidator in exchange for the collateral depends on the outcome of a 2-phase auction. If the auction is perfectly \emph{efficient}, the winning bidder pays off an equivalent amount of debt for receiving the offered collateral. On the other hand, when the auction is inefficient due to system congestion, collusion, transaction censoring, etc., the winning bidder can receive the entire collateral on offer without paying off an equivalent amount of debt. In our simplified Contract $\makercon^{(\mathbf{{X}, \mathbf{Y}})}$, we assume that liquidation is perfectly efficient.

\mypara{Uniswap as a price oracle for Maker.}
If Uniswap is used as a price oracle in the Maker contract, by reordering Uniswap transactions, and thereby manipulating the exchange rate, a miner can cause the value of a user's collateral to fall below the acceptable threshold, and trigger a liquidation event. Furthermore, the miner can buy the user's collateral tokens in the liquidation event, and later sell them for a profit when the exchange price returns to normal.
\fi

\iffull \else 

\fi


\iffull\else
\section{Uniswap CFF Model Details}
\label{sec:moreuni}

\begin{figure}[!t]
\fpage{0.9}{
    \begin{center}
        \textbf{Contract $\uniswapcon^{(\mathbf{X}, \mathbf{Y})}$}
    \end{center}
    
     $\functioncode \texttt{exchange(\textrm{InToken, OutToken, InAmount}):}$ \\
    \myind \ifcode $\accbalance(\calleracc)[\textrm{InToken}] \geq \textrm{InAmount}$ \thencode \\
    \myind \myind $x = \accbalance(\uniswapcon)[\textrm{InToken}]$ \\
    \myind \myind $y = \accbalance(\uniswapcon) [\textrm{OutToken}]$\\
    \myind \myind $\textrm{OutAmount} = y - {xy}/{(x + \textrm{InAmount})}$\\
    \myind \myind $\accbalance(\calleracc)[\textrm{InToken}] \minuseq \textrm{InAmount}$\\
    \myind \myind $\accbalance(\calleracc)[\textrm{OutToken}] \pluseq \textrm{OutAmount}$\\
    \myind \myind $\accbalance(\uniswapcon)[\textrm{InToken}] \pluseq \textrm{InAmount}$\\
    \myind \myind $\accbalance(\uniswapcon)[\textrm{OutToken}] \minuseq \textrm{OutAmount}$
    
    \myind \elsecode Output $\bot$
    
}
\caption{Simplified abstract Uniswap contract}
\label{fig:uniswap_contract}
\iffull \else \vspace{-2em} \fi
\end{figure}

\begin{figure*}[t]
    \lstset{language=C}
    \begin{lstlisting}
    claim <k>
         On UniswapV2 697323163401596485410334513241460920685086001293 swaps for ETH by providing 1300000000000000000000 COMP and 0 ETH with change 0 fee 1767957155464 ;
         On Sushiswap Miner swaps for ETH by providing Alpha:Int COMP and 0 ETH with change 0 fee 0 ;
         On UniswapV2 Miner swaps for Alpha COMP by providing ETH fee 0 ;
    
         => .K
         </k>
         <S> (Sushiswap in COMP) |-> 107495485843438764484770 (Sushiswap in ETH) |-> 49835502094518088853633 (UniswapV2 in COMP) |-> 5945498629669852264883 (UniswapV2 in ETH) |-> 2615599823603823616442 => ?S:Map </S>
         <B> .List => ?_ </B>
         requires (Alpha >Int 0) andBool (Alpha <Int 10000000000000000000000) //10**22
         ensures ({?S[Miner in COMP]}:>Int <=Int 0  ) andBool ({?S[Miner in ETH]}:>Int <=Int 0  )
    
         
    \end{lstlisting}
    \iffull \else \vspace{-2em} \fi
    \caption{Specification for Composition of Sushiswap and UniswapV2}
    \label{fig:kproof_spec}
    \iffull \else \vspace{-2em} \fi
    \end{figure*}
We detail simplified Uniswap and Maker contracts 
in \figs~\ref{fig:uniswap_contract} and~\ref{fig:maker_contract}.

\fi

\section{CFF Details}

\subsection{Why K?}
\label{subsec:considerations}
A natural question is why we chose the K framework for our implementation of the CFF. While CFF can be instantiated using any good formal verification tool, we found K code to be especially human readable and intuitive (mainly because of its concurrent semantics) for developers who may not be experts in formal verification. Prior work~\cite{hildenbrandt2018kevm} has already implemented full EVM semantics using K. We also chose K for qualitative reasons, detailed in Section~\ref{subsec:formal-verification}. We emphasize that our results are not tool-specific, and should be straightforward to replicate. 

\mypara{K vs.~Coq.}
As a specific comparison point, we explain our choice of K over Coq~\cite{leroy2009formal}, another popular formal verification tool. A comparison in~\cite{kvcoq} found similar performance numbers for the proving engines of both K and Coq; simple proofs took approximately the same amount of real time on test hardware. We posit (though defer detailed study) that performance differences would be minor. As~\cite{kvcoq} points out, however, models in K are always executable, and allow for concrete inputs to be evaluated. On the other hand, in Coq, execution must be defined separately as its own function and proved equivalent to the relational definition of the corresponding models. We believe that this additional step would impose substantial overhead on model development our framework.

\iffull
\subsection{Writing a CFF Model}
\label{sec:morek}

We now provide additional description of the operations executed by our model in Figure~\ref{fig:uniswap_k_contract}, which may prove helpful when defining your own CFF model.

Line 2 is the first such operation, and creates a local variable in the model state which binds the amount to send according to the AMM formula in a variable called AmountToSend.  USwapBalanceIn and USwapBalanceOut, the balance of the Uniswap contract in the input and output tokens, are used in this calculation.  These variables are sourced from Line 9, where they are matched in the global Ethereum state S.

Lines 3-6 use the special ``gets" operator, which we give operational semantics to separately, to change the system state by debiting and crediting the appropriate balances from Uniswap and the user in the traded asset; the user here receives AmountToSend tokens for TradeAmount tokens sent as input.  The var function is a built-in function indicating that a variable bound in the current scope (rather than the Ethereum state) should be used.

Note some special K keywords are required for our semantic rule.  Firstly, the \texttt{...} keyword specifies that anything can successfully match in this location when applying the rule.  We do not care what operations after the first execution are pending in the model when applying the rule, for example (Line 7).  The rule also applies regardless of the contents of S outside of USwapBalance (Line 9), and regardless of what transactions the miner has already included in their block B in the model (Line 10).

Further such details on K can be found at the K tutorial at \url{http://kframework.org/}, or by reading our modeling code and documentation on Github.

\iffull 

\fi

\subsection{Refinements to the Abstract Model}
\label{sec:refinements}
\subsubsection{AMM Refinements}
We now refine the abstract model, intended to illustrate the core functionality of a Uniswap-like AMM, to be fully faithful to the deployed Uniswap contract.  To do so, we must refine our trade rule to take into account the rounding used in the real-world Uniswap trade functions, which come with some degree of error/imprecision.  This imprecision is described and formalized in the model in~\cite{verifiedsc}, a superset of the formal semantics in our work that was used to verify the Uniswap protocol before deployment.

We must also add semantic rules for liquidity provision and removal transactions, which further affect the Uniswap contract and drive relevant state updates.  Lastly, we must take into account all code paths in the Uniswap deployed Solidity contract.  Our fully refined model that accurately reflects real-world Uniswap arithmetic is available in \texttt{models/uniswap} on our Github repository.

\subsubsection{MakerDAO Refinements and Liquidation Auction}
We refine our abstract MakerDAO model by adding a rule to update stability fees and account for this stability fees to calculate the CDP debt accurately. We also combine the CDP manipulation actions into one single rule, to accurately reflect the deployed contract. Next, we add rules for CDP fungibility, i.e. transferring debt and collateral between CDPs. We now discuss the subtleties of MEV extraction in a liquidation auction and replace the efficient auction outcome in our abstract model accordingly.

We analyze the optimum MEV assuming that all network miners are behaving to maximize total MEV -- a rational decision from a miner standpoint. Optimum MEV is achieved when the miner is able to censor competing bids, and win the entire collateral on offer in the second phase of the auction. We thus refine the liquidation auction outcome in our abstract model to receive the entire collateral on offer. Note that if some miners defect to reduce the efficiency of MEV extraction, it is possible that only some constant percentage of the optimum MEV will remain extractable.
\fi
\subsection{Mechanized Proofs}
\label{sec:kproof}
\begin{figure}
\centering
\includegraphics[scale=0.45]{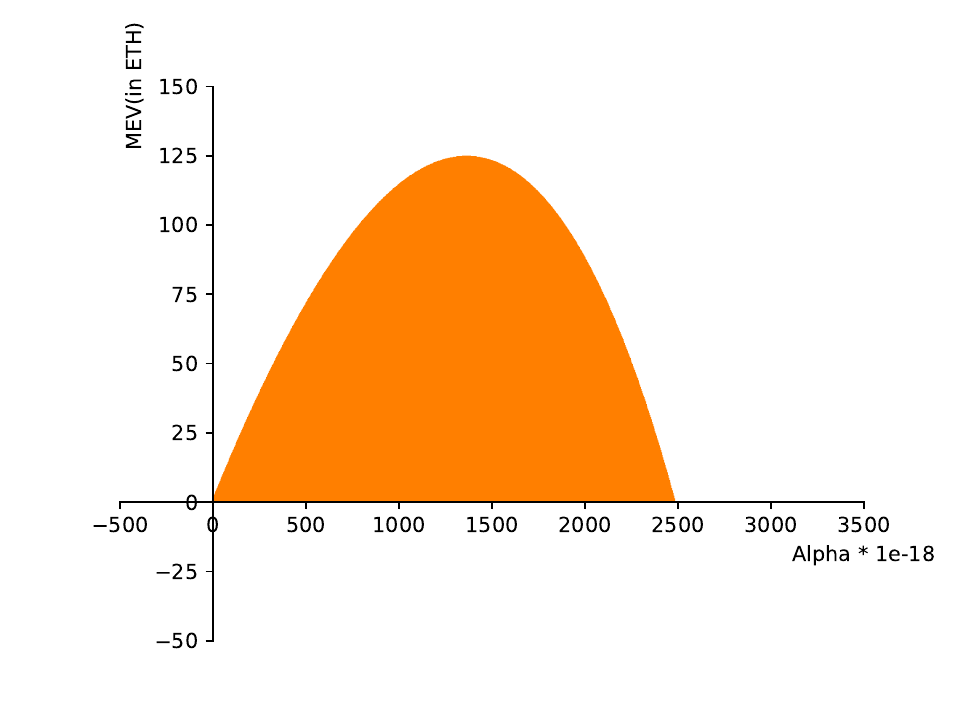}
\iffull \else \vspace{-1.5em} \fi
\caption{The region boundary represents MEV extractable by the miner as a function of the input variable (size of its trade). The maximum value is 123 ETH.}
\label{fig:counterexample}
\iffull \else \vspace{-2em} \fi
\end{figure}
We now provide the details of an example specification used to check the security of the composition of Sushiswap and UniswapV2 in \fig~\ref{fig:kproof_spec}. This example captures one of the biggest arbitrage captures\footnote{\href{https://etherscan.io/tx/0x2c79cdd1a16767e90d55a1598c833f77c609e972ea0fa7622b70a67646a681a5}{https://etherscan.io/tx/0x2c79cdd1a16767e90d55a1598c833f77c609e97\\2ea0fa7622b70a67646a681a5}} observed on-chain involving two AMMs as reported in~\cite{qin2021quantifying}. The hex addresses for users are converted to base 10 integers. The initial state for Sushiswap and UniswapV2 is specified in the \scell cell. The last two transactions in the \kcell cell represent the transactions inserted by the Miner according to the strategy described in Section~\ref{sec:composability}. Note that the Miner transaction can be symbolic, $Alpha$ being the symbol representing the size of the swap Miner does denominated in Wei (1 ETH = 1e18 Wei). The \texttt{requires} clause specifies the constraints on $Alpha$, essentially denoting the Miner budget. Finally, the \texttt{ensures} clause represents the claim that the Miner is not able to extract any value regardless of the way specified transactions are reordered.

Our tool derives a counterexample to the claim with the MEV formula given by (plotted in \fig~\ref{fig:counterexample}):

{\footnotesize \begin{verbatim}
-1 - 2147460244936306246609000 * Alpha / ( 997 * 
( 7245498629669852264883 - Alpha ) ) + 997 * Alpha 
* 49835502094518088853633 / ( 997 * Alpha + 
107495485843438764484770000)
\end{verbatim}
}

\subsection{Bounding the MEV for AMMs}
\label{sec:ammbounds}

Although the price offered by the AMMs we study at the end of a given set of transactions is independent of the order of the transactions~\cite{buterin2017path}, individual users' transactions get different prices depending on the order of the transactions. A miner can thus influence the value individual users get for their trades by choosing a different order for the transactions. For each user, there is an optimal and a worst case ordering. 

Let $b_h$ be the highest ETH-value of a trader's account after a block has elapsed, assuming access to a price oracle for pricing a user's tokens at an invariant market price for the time of trade execution.  Let $b_l$ be the lowest such value.  It is therefore rational for the trader to pay $b_h - b_l - \epsilon$ to miners as a bribe.  For miners to elicit this bribe, they would deploy a contract allowing each user of an AMM to deposit ETH.  They would then credibly commit to mining the order resulting in $b_l$ if no funds were available.  Otherwise, they would submit both $b_l$ and $b_h$, along with associated proofs, to the smart contract, which would enforce the order resulting in $b_h$, pay the miner $b_h - b_l - \epsilon$, and pay the trader $\epsilon$.  Note that paying such a contract is a strictly dominant from a trader point of view, as the trader profits $\epsilon$ more than without paying into such a contract. Introducing this new contract increases MEV by exactly $b_h - b_l - \epsilon$ through a direct payment by inspection; in our experiments, we assume $\epsilon$ is negligible when compared to $b_h - b_l$: since being paid this $\epsilon$ is a strictly dominant strategy, miners need only compensate users for the low cost of locking capital (which can be removed freely) in the bribery contract.

When analyzing attacks like this on DeFi protocols, a natural question becomes how to efficiently and thoroughly uncover reordering-based differences that would allow for an accurate measurement of $b_h$ and $b_l$, and therefore the MEV in the presence of such contract composition.  It is this measurement on which we focus in our AMM experiments.

\iffull
\subsection{MEV Deep Dive}
\label{sec:deepdive}

In this section, we will explore in more detail our top MEV example, which occurred in Ethereum block 10968577 in the YFI-WETH pair on Sushiswap, primarily surrounding MEV-creating transaction {\small 0x8a9d88084eb3a451fcd1c28f1851d0-ced03e7665499a362942978ff13d5c19d4}.

In this transaction, a user sold 40 YFI tokens, a popular and extremely valuable Ethereum token that was in the middle of an upwards price rally, on an automated decentralized exchange liquidity aggregator called 1inch.exchange.  As part of this aggregation, 1inch chose to execute a sale of 22 YFI tokens on SushiSwap, worth USD\$550,000 at the then-price of USD\$25,000 per token.

\begin{figure}[h!]
\centering
\begin{tabular}{|c|c|c|}
\toprule

\textbf{User} & \textbf{Swap Performed} & \textbf{Amount of Input Token} \\
\midrule
\textbf{A} & \color{blue}YFI\color{black} $\rightarrow$ \color{orange}WETH\color{black} & 22000000000000000000\\
\midrule
\textbf{B} & \color{orange}WETH\color{black} $\rightarrow$ \color{blue}YFI\color{black} & 53788258395569781028 \\
\midrule
\textbf{B} & \color{orange}WETH\color{black} $\rightarrow$ \color{blue}YFI\color{black} & 6784028349336991312 \\
\midrule
\textbf{C} & \color{orange}WETH\color{black} $\rightarrow$ \color{blue}YFI\color{black} & 103266050000000000000 \\
\midrule
\textbf{D} & \color{orange}WETH\color{black} $\rightarrow$ \color{blue}YFI\color{black} & 300000000000000000000 \\
\midrule
\textbf{D} & \color{orange}WETH\color{black} $\rightarrow$ \color{blue}YFI\color{black} & 4970140364366149478 \\
\midrule
\textbf{D} & \color{orange}WETH\color{black} $\rightarrow$ \color{blue}YFI\color{black} & 6984067876806377830 \\
\midrule
\textbf{D} & \color{orange}WETH\color{black} $\rightarrow$ \color{blue}YFI\color{black} & 300000000000000000000 \\
\midrule
\textbf{D} & \color{orange}WETH\color{black} $\rightarrow$ \color{blue}YFI\color{black} & 150000000000000000000 \\
\bottomrule
\end{tabular}

\caption{Mined actual transaction ordering of the top MEV block in our sample} \label{figure:topmev}

\end{figure}

Because the user placed a large market order on a set of automated market makers, this naturally created an arbitrage opportunity to buy YFI at this newly-depressed price, selling it into more liquid off-chain and on-chain markets which still reflected the real market valuation.  Figure~\ref{figure:topmev} shows the ordering of transactions on the network, with user ``A" being the user selling YFI tokens on 1inch, and users B-D representing a set of arbitrage bots that programatically bought and re-sold tokens from SushiSwap when user A created an arbitrage opportunity.

The MEV here is obvious, as the ability for a miner to essentially take the trades performed by the bots furnishes a more profitable opportunity for the miner than the bots, who can also execute the bots' failed transactions.

The optimal ordering found by our tool for user A is D $\rightarrow$ D $\rightarrow$ D $\rightarrow$ C $\rightarrow$ B $\rightarrow$ B $\rightarrow$ D $\rightarrow$ D $\rightarrow$ A, where the user's trade is executed after the trade of the arbitrage bots.  This makes sense, as the arbitrage bots cannot take advantage of the user's price movement to re-arbitrage Uniswap back to market parity.  Conversely, the worst order for user A is A $\rightarrow$ D $\rightarrow$ D $\rightarrow$ D $\rightarrow$ C $\rightarrow$ B $\rightarrow$ D $\rightarrow$ B $\rightarrow$ D, which is very similar to the order actually mined by the miner.

Note that the arbitrage on the network are already somewhat effective at extracting MEV from Uniswap, a result that is expected given the conclusions of~\cite{daian2020flashboys}.  However, the miner can still increase profit even further over these bots, due to the fine-grained control it can exercise in ordering that is likely hard to achieve through the public priority gas auctions described in~\cite{daian2020flashboys}.  We thus posit that this example shows not only the existence of MEV that can be exploited through generic tooling, but also the relative inefficiencies of current arbitrage bots on the network, who are unable to achieve the maximally optimal order even when opportunity sizes are large.

\fi



\end{document}